\documentclass[11pt]{article}

\usepackage[margin=1.12in]{geometry}
\usepackage{amsmath,amssymb,amsfonts,amsthm,mathtools,bm}
\usepackage{booktabs}
\usepackage{graphicx}
\usepackage{enumitem}
\usepackage[round,authoryear]{natbib}
\usepackage{algorithm}
\usepackage{algorithmic}
\usepackage{color}
\usepackage[hidelinks]{hyperref}
\mathtoolsset{showonlyrefs=true}
\allowdisplaybreaks
\usepackage{mathtools}
\mathtoolsset{showonlyrefs=true}
\newtheorem{theorem}{Theorem}[section]
\newtheorem{lemma}{Lemma}[section]
\newtheorem{proposition}{Proposition}[section]
\newtheorem{corollary}{Corollary}[section]
\theoremstyle{definition}
\newtheorem{assumption}{Assumption}[section]

\newcommand{\R}{\mathbb R}
\newcommand{\E}{\mathbb E}
\newcommand{\mE}{\mathbf{E}}
\newcommand{\Prb}{\mathbb P}
\newcommand{\Var}{\operatorname{Var}}
\newcommand{\tr}{\operatorname{tr}}
\newcommand{\diag}{\operatorname{diag}}
\newcommand{\argmin}{\operatorname*{arg\,min}}

\newcommand{\T}{^{\top}}
\newcommand{\dd}{\stackrel{d}{\longrightarrow}}
\newcommand{\norm}[1]{\lVert #1\rVert}
\newcommand{\abs}[1]{\left|#1\right|}
\newcommand{\sym}{\operatorname{sym}}
\newcommand{\ind}{\mathbf 1}
\newcommand{\calX}{\mathcal X}
\newcommand{\Estar}{\mathbb E^*}
\newcommand{\Prstar}{\mathbb P^*}
\newcommand{\Varstar}{\operatorname{Var}^*}
\newcommand{\qhat}{\widehat q}

\newcommand{\vX}{\bm{X}}
\newcommand{\vx}{\bm{x}}
\newcommand{\vt}{\bm{t}}
\newcommand{\vtheta}{\bm{\theta}}
\newcommand{\vu}{\bm{u}}
\newcommand{\vv}{\bm{v}}
\newcommand{\vY}{\bm{Y}}
\newcommand{\vS}{\bm{S}}
\newcommand{\vZ}{\bm{Z}}
\newcommand{\ve}{\bm{e}}
\newcommand{\vg}{\bm{g}}
\newcommand{\vh}{\bm{h}}
\newcommand{\va}{\bm{a}}
\newcommand{\vdelta}{\bm{\delta}}
\newcommand{\vDelta}{\bm{\Delta}}
\newcommand{\vr}{\bm{r}}
\newcommand{\mI}{\mathbf{I}}
\newcommand{\mD}{\mathbf{D}}
\newcommand{\mR}{\mathbf{R}}
\newcommand{\mA}{\mathbf{A}}
\newcommand{\mB}{\mathbf{B}}
\newcommand{\mC}{\mathbf{C}}
\newcommand{\mG}{\mathbf{G}}
\newcommand{\mK}{\mathbf{K}}
\newcommand{\mL}{\mathbf{L}}
\newcommand{\mM}{\mathbf{M}}
\newcommand{\mOmega}{\mathbf{\Omega}}
\newcommand{\mH}{\mathbf{H}}
\newcommand{\mX}{\mathbf{X}}

\title{High-Dimensional Two-Sample Test for Elliptical Symmetry Distribution}
\author{Long Feng$^{1}$ and Hongfei Wang$^{2}$\\
$^{1}$Nankai University\\
 $^{2}$Nanjing Audit University}
\date{}
\begin{document}
\maketitle

\begin{abstract}
We study the high-dimensional two-sample location problem under elliptical symmetry with arbitrary dependence in the scatter matrix. Existing spatial-sign procedures are attractive for heavy-tailed data, but their null calibration is tied to weakly dependent scatter matrices and their diagonal standardization does not, in general, recover the diagonal shape under strong dependence. We propose a new spatial-sign test based on coordinatewise pairwise-difference quantile scales. The new diagonal standardizer is location free, requires no positive moment condition on the radial variable, and estimates the diagonal of the elliptical shape up to a {\color{black}scalar specific to the sample}, which disappears after spatial normalization. For the resulting full-sample statistic{\color{black},} we derive an explicit-rate stochastic expansion, establish a general weighted chi-square null distribution under arbitrary correlation structure, justify an empirical diagonal-deletion correction, and show that a Rademacher wild bootstrap consistently estimates the null law. The usual normal approximation appears only as a special case when no eigenvalue dominates.
\end{abstract}


\medskip
\noindent\textbf{MSC2020 subject classifications:} Primary 62H15; secondary 62G35, 62E20.

\section{Introduction}

Testing the equality of two high-dimensional mean or location vectors is a fundamental problem in modern multivariate analysis.  Applications arise in genomics, signal processing, neuroimaging, finance, and many other areas where the dimension $p$ is comparable to or much larger than the sample sizes.  In low dimension, the benchmark procedure is Hotelling's $T^2$ test \citep{hotelling1931}.  When the dimension is larger than the sample size, however, the pooled sample covariance matrix is singular and the classical Mahalanobis form is no longer available.  This difficulty has generated a large literature on nonclassical tests for high-dimensional means.  For broad reviews, see \citet{huangliyang2022} and \citet{harrarkong2022}.

Many high-dimensional mean tests avoid estimating or inverting the full covariance matrix by replacing the Mahalanobis distance with the Euclidean distance, leading to quadratic-form procedures that are particularly suited to Gaussian or light-tailed data. \citet{bai1996} showed that the squared Euclidean distance between the sample means continues to yield a useful test when Hotelling's statistic breaks down.  \citet{chenqin2010} removed the self-inner-product bias and obtained an unbiased $U$-statistic with strong performance against dense alternatives.  \citet{huetal2017} extended this quadratic-form strategy to unequal covariance matrices.  

Another important approach improves scale invariance by standardizing each coordinate with an estimate of its marginal scale, thereby reducing the influence of heterogeneous variances across dimensions. \citet{srivastavadu2008} proposed a diagonal Hotelling-type statistic for the one-sample problem.  \citet{srivastava2013} adapted this idea to the two-sample Behrens--Fisher setting with unequal covariance matrices.  \citet{parkayyala2013} introduced leave-out corrections to reduce plug-in bias. \citet{FengZouWangZhu2015} further studied the high-dimensional Behrens--Fisher problem and proposed a scale-invariant test that accommodates unequal covariance matrices across the two samples, thereby extending two-sample mean inference beyond the equal-covariance setting.  More recently, \citet{zhang2020scale} proposed a simple scale-invariant two-sample test and showed that, after normalization, its limit can be either normal or nonnormal, depending on the covariance structure.
For sparse alternatives under dependence, \citet{cailiuxia2014} proposed a {\color{black}max}-type procedure that exploits the largest standardized coordinatewise discrepancies. Considerable effort has also been devoted to constructing adaptive procedures that remain powerful over a broad range of sparsity regimes. \citet{XuLinWeiPan2016} proposed an adaptive two-sample test that combines a family of \(L_p\)-type statistics and thereby bridges dense and sparse alternatives. More recently, \citet{FengJiangLiLiu2024} established a general {\color{black}theory of asymptotic independence between} the sum-type and max-type statistics under weak dependence, which provides a rigorous foundation for adaptive combination tests that remain effective across different sparsity patterns.

A closely related body of work is particularly relevant to the present paper because it explicitly allows arbitrary covariance structures.  \citet{zhangxu2009} studied the $k$-sample Behrens--Fisher problem by an $L^2$-norm statistic together with chi-square approximation.  \citet{zhang2021normal} developed a normal-reference approach for high-dimensional two-sample Behrens--Fisher problems and showed that the test statistic and a chi-square-type mixture share the same normal or nonnormal limit under general covariance heterogeneity.  \citet{wangxu2022} proposed an approximate randomization test based on the {\color{black}statistic in \cite{chenqin2010}} and established validity without imposing conditions on the eigenstructure of the covariance matrices.  \citet{zhang2023nrscale} further developed a normal-reference scale-invariant test, while \citet{zhuwangzhang2024} proposed an $F$-type normal-reference procedure for the same Behrens--Fisher setting.  These papers make clear that, under strong dependence, the correct reference law for high-dimensional mean tests is often a chi-square-type mixture rather than a Gaussian limit.

Robust procedures are needed when the underlying distributions are heavy-tailed, contaminated, or only weakly moment constrained, because mean-based quadratic statistics can be highly sensitive to outliers and may suffer from severe size distortion in finite samples. This has motivated a substantial literature on high-dimensional mean testing based on ranks, signs, spatial signs, and other robust transformations \citep{Feng2026EllipticalHDDA}. \citet{wangpengli2015} proposed a high-dimensional nonparametric test for the mean vector under general multivariate distributions.  \citet{fengsun2016} developed a scalar-invariant one-sample spatial-sign test under elliptical symmetry.  For the two-sample problem, \citet{feng2016} introduced a leave-one-out spatial-sign procedure, and \citet{li2016} proposed a simpler bias-corrected version in the same framework.  \citet{wanglintang2019} studied a robust two-sample mean test under dependence, and \citet{jiangwangwen2024} recently proposed nonparametric two-sample tests via random integration.  \citet{chakrabortychaudhuri2017} compared mean-, sign-, and rank-based procedures and clarified how their relative efficiencies depend on tail behavior and the geometry of the alternative.

Despite this substantial literature, there remains a gap between the arbitrary-covariance mean-based line and the heavy-tail-robust spatial-sign line.  The two-sample spatial-sign tests of \citet{feng2016} and \citet{li2016} are calibrated under trace conditions that are naturally compatible with sparse or sufficiently weakly dependent scatter matrices.  When the scatter matrix is strongly and densely correlated, two difficulties arise simultaneously.  First, the leading quadratic form retains the whole spectrum of the scatter matrix, so the null law is generally a weighted chi-square mixture rather than a normal distribution.  Second, the fixed-point diagonal standardizer used in \citet{feng2016} and \citet{li2016} does not, in general, recover the diagonal of the elliptical shape matrix under arbitrary dependence.  Hence the existing spatial-sign theory becomes unreliable precisely in the regime where the normal-reference mean-based literature predicts nonnormal limits.

This paper resolves these two issues under elliptical symmetry.  We replace the spatial-sign-based fixed-point diagonal standardizer by a coordinatewise pairwise-difference quantile estimator.  The new diagonal estimator is location free, requires no positive moment condition on the radial variable, and recovers the diagonal shape up to a within-sample scalar that disappears after spatial normalization.  Based on this diagonal standardizer and full-sample spatial medians, we construct a diagonal-deleted spatial-sign statistic whose leading term is a canonical quadratic statistic.  Under arbitrary correlation structure, we derive explicit-rate expansions, establish a general weighted chi-square null distribution, justify an empirical diagonal-deletion correction, and show that a Rademacher wild bootstrap consistently estimates the null law.

The rest of the paper is organized as follows.  Section~2 introduces the proposed procedure and states the main theoretical results.  Section~3 reports simulation results under strong dependence.  Appendix~A contains the auxiliary {\color{black}lemmas} and the proofs of all theoretical results, while Appendix~B discusses the radius and spectral assumptions in representative elliptical models.

\paragraph{Notation.}
For a vector \(\bm{x}\), let \(\|\bm{x}\|\) denote its Euclidean norm, and define
\(U(\bm{x})=\bm{x}/\|\bm{x}\|\) for \(\bm{x}\neq \bm{0}\) and \(U(\bm{0})=\bm{0}\).
For a matrix \(\mathbf{A}\), write \(\|\mathbf{A}\|_{\rm op}\) and
\(\|\mathbf{A}\|_{F}\) for its operator and Frobenius norms, respectively,
\(\operatorname{tr}(\mathbf{A})\) for its trace, and
\(\lambda_{\max}(\mathbf{A})\) and \(\lambda_{\min}(\mathbf{A})\) for its largest
and smallest eigenvalues when \(\mathbf{A}\) is symmetric.  For a vector
\(\bm{a}\), \(\operatorname{diag}(\bm{a})\) denotes the diagonal matrix with
diagonal entries given by \(\bm{a}\); for a square matrix \(\mathbf{A}\),
\(\operatorname{diag}(\mathbf{A})\) denotes the diagonal matrix formed from the
diagonal entries of \(\mathbf{A}\).  {\color{black}Let \(\mathbf{I}_k\) be the \(k\times k\)
identity matrix and \(\bm{1}_k=(1,\ldots,1)^{\top}\in\mathbb R^k\) for any $k\in \mathbb Z^{+}$.}
We write \(a_n\lesssim b_n\) if \(a_n\le Cb_n\) for a universal constant
\(C>0\), and \(a_n\asymp b_n\) if both \(a_n\lesssim b_n\) and
\(b_n\lesssim a_n\).  %The symbols \(O_p(\cdot)\) and \(o_p(\cdot)\) have their usual meanings.  
Conditional probability and expectation with respect to the
wild-bootstrap multipliers are denoted by \(\Pr^*(\cdot)\) and \(\mathbb E^*(\cdot)\).
%Throughout, {\color{black}we let $n_1$ and $n_2$ be the sample sizes, $p$ the dimension, and define} \(N=n_1+n_2\), \(n_0=\min(n_1,n_2)\), and
%\(x_N=\log(4pN)\).

\section{Method}

We consider two independent samples
$
\{\vX_{1i}\}_{i=1}^{n_1}$
and
$\{\vX_{2i}\}_{i=1}^{n_2},$
drawn from two $p$-variate elliptically symmetric populations. Specifically, for $k=1,2$ and $i=1,\ldots,n_k$, we assume that
\begin{equation}\label{eq:model}
\vX_{ki}
=
\vtheta_k
+
(\mD_k^\circ)^{1/2}\mR_k^{1/2}\xi_{ki}\vu_{ki},
\end{equation}
where $\vtheta_k\in\mathbb{R}^p$ is the location parameter of the $k$th population,
$
\mD_k^\circ=\diag(d_{k1}^\circ,\ldots,d_{kp}^\circ)
$
is a diagonal matrix with strictly positive entries, $\mR_k$ is a positive definite correlation matrix, $\vu_{ki}\sim \mathrm{Unif}(\mathbb{S}^{p-1})$, and $\vu_{ki}$ is independent of the positive radial variable $\xi_{ki}$. Throughout the paper, no normalization condition such as $\E(\xi_{ki}^2)=p$ is imposed. Under model \eqref{eq:model}, the two populations are allowed to have arbitrary correlation structures through $\mR_1$ and $\mR_2$.

Our goal is to test whether the two populations have the same location vector, that is,
\begin{equation}\label{eq:hypothesis}
H_0:\ \vtheta_1=\vtheta_2
\qquad\text{versus}\qquad
H_1:\ \vtheta_1\neq\vtheta_2.
\end{equation}
\citet{feng2016} proposed the following leave-one-out statistic
\begin{equation}\label{eq:FZW}
R_n^{\mathrm{FZW}}
=
-\frac{1}{n_1n_2}\sum_{i=1}^{n_1}\sum_{j=1}^{n_2}
U\{\widetilde{\mD}_{1,i}^{-1/2}(\vX_{1i}-\widetilde{\vtheta}_{2,j})\}\T
U\{\widetilde{\mD}_{2,j}^{-1/2}(\vX_{2j}-\widetilde{\vtheta}_{1,i})\},
\end{equation}
where \((\widetilde{\vtheta}_{k,\ell},\widetilde{\mD}_{k,\ell})\) are computed by the following algorithm \ref{alg:feng-iteration} from the \(k\)th sample with observation \(\ell\) deleted.  

\begin{algorithm}[htbp]
\caption{Iterative estimation of $\bm{\theta}_i$ and $\mathbf{D}_i$ in \cite{feng2016}}
\label{alg:feng-iteration}
\begin{algorithmic}[1]
\STATE Choose initial values $\bm{\theta}_i^{(0)}$ and $\mathbf{D}_i^{(0)}$.
\STATE Set $t=0$.
\REPEAT
\STATE Compute
$$
\bm{\epsilon}_{ij}^{(t)}
=
\bigl(\mathbf{D}_i^{(t)}\bigr)^{-1/2}
\bigl(\mathbf{X}_{ij}-\bm{\theta}_i^{(t)}\bigr),
\qquad
j=1,\ldots,n_i.
$$
\STATE Update $\bm{\theta}_i$ by
$$
\bm{\theta}_i^{(t+1)}
=
\bm{\theta}_i^{(t)}
+
\frac{
\bigl(\mathbf{D}_i^{(t)}\bigr)^{1/2}
\sum_{j=1}^{n_i} U\!\left(\bm{\epsilon}_{ij}^{(t)}\right)
}{
\sum_{j=1}^{n_i}\left\|\bm{\epsilon}_{ij}^{(t)}\right\|^{-1}
}.
$$
\STATE Update $\mathbf{D}_i$ by
$$
\mathbf{D}_i^{(t+1)}
=
p\,
\bigl(\mathbf{D}_i^{(t)}\bigr)^{1/2}
\operatorname{diag}
\left\{
n_i^{-1}
\sum_{j=1}^{n_i}
U\!\left(\bm{\epsilon}_{ij}^{(t)}\right)
U\!\left(\bm{\epsilon}_{ij}^{(t)}\right)^{\top}
\right\}
\bigl(\mathbf{D}_i^{(t)}\bigr)^{1/2}.
$$
\STATE Set $t\leftarrow t+1$.
\UNTIL{a convergence criterion is satisfied}
\end{algorithmic}
\end{algorithm}

\citet{li2016} proposed a simpler bias-corrected statistic:
\begin{equation}\label{eq:LWZ}
\begin{aligned}
T_n^{\mathrm{LWZ}}
={}&
-\frac{1}{n_1n_2}\sum_{i=1}^{n_1}\sum_{j=1}^{n_2}
U\{\widetilde{\mD}_{1}^{-1/2}(\vX_{1i}-\widetilde{\vtheta}_{2})\}\T
U\{\widetilde{\mD}_{2}^{-1/2}(\vX_{2j}-\widetilde{\vtheta}_{1})\}\\
&\quad
-\frac{\hat c_2}{\hat c_1 n_1p}\tr(\widetilde{\mD}_{1}^{1/2}\widetilde{\mD}_{2}^{-1/2})
-\frac{\hat c_1}{\hat c_2 n_2p}\tr(\widetilde{\mD}_{2}^{1/2}\widetilde{\mD}_{1}^{-1/2}),
\end{aligned}
\end{equation}
where
\begin{equation}\label{eq:chat}
    \hat c_k=\frac{1}{n}\sum_{i=1}^{n}
    \norm{\widetilde{\mD}_{k}^{-1/2}(\vX_{ki}-\widetilde{\vtheta}_{k})}^{-1},
    \qquad k=1,2.
\end{equation}

Both statistics are calibrated by normal-reference theory under trace conditions that effectively require the scatter matrices to be sparse or sufficiently weakly dependent.  Under strong dense dependence, however, the leading quadratic form retains the whole spectrum of the scatter matrix, and the null law is generally a weighted chi-square mixture rather than a normal distribution.  Mean-based procedures under arbitrary covariances have been developed by \citet{zhang2021normal,wangxu2022,zhang2023nrscale,zhuwangzhang2024}, but those results do not transfer directly to \eqref{eq:FZW}--\eqref{eq:LWZ} because spatial signs require an additional diagonal standardization and the fixed-point diagonal estimator itself fails under general elliptical dependence.  In particular, the sign-based fixed-point diagonal standardization used in \citet{feng2016} and \citet{li2016} does not, in general, target the diagonal of the elliptical shape under arbitrary dependence.  This is the second difficulty addressed here.

We therefore estimate the diagonal shape coordinatewise through pairwise-difference quantiles. 
Let \(\vv\in\mathbb S^{p-1}\) {\color{black} and $\bm e_1$
 be the p-dimensional unit vector with 1 in the 1-th position and 0 elsewhere. }.  Since \(\vu_{ki}\) is uniform on the sphere, the law of \(\xi_{ki}\vv\T \vu_{ki}\) is independent of \(\vv\).  For independent copies \(Z_{k1}\) and \(Z_{k2}\) of \(\xi_{ki}\ve_1\T \vu_{ki}\), define
\[
        F_{k,\Delta}(t)=\Prb\{\abs{Z_{k1}-Z_{k2}}\le t\},
        \qquad
        q_{k,\alpha}=F_{k,\Delta}^{-1}(\alpha),
\]
and set
\begin{equation}\label{eq:Dtarget}
        \mD_k=q_{k,\alpha}^2\mD_k^\circ{\color{black}=(d_{k1},\dots,d_{kp})\T},
        \qquad
        d_{kj}=q_{k,\alpha}^2d^\circ_{kj}.
\end{equation}
The constants \(q_{k,\alpha}\) are nuisance scale factors.  They are immaterial for spatial signs because \(U\{(c\mD)^{-1/2}\vx\}=U(\mD^{-1/2}\vx)\) for every \(c>0\).

Define the oracle standardized variables and signs
\begin{equation}\label{eq:YS}
        \vY_{ki}=\mD_k^{-1/2}(\vX_{ki}-\vtheta_k),
        \qquad
        \vS_{ki}=U(\vY_{ki}).
\end{equation}
Then
\[
        \vS_{ki}=\frac{\mR_k^{1/2}\vu_{ki}}{(\vu_{ki}\T \mR_k\vu_{ki})^{1/2}},
\]
so the sign distribution depends only on the angular part \(\mR_k\), not on the radial law.  Put
\begin{align}
    \mOmega_k&=\E(\vS_{ki}\vS_{ki}\T),
    \tr(\mOmega_k)=1,\label{eq:Omega}\\
    \mG_k&=\E\left[\norm{\vY_{ki}}^{-1}\{\mI_p-\vS_{ki}\vS_{ki}\T\}\right].\label{eq:G}
\end{align}
The matrix \(\mG_k\) is the population Jacobian of the spatial-median score under the working diagonal standardization. For a fixed \(\alpha\in(0,1)\), define {\color{black}$\bm X_{ki}=(X_{ki1},\dots ,X_{kip})\T$,} 
\begin{equation}\label{eq:pdq-cdf}
    \widehat F_{kj}(t)
    =\binom{n_k}{2}^{-1}\sum_{1\le i<\ell\le n_k}
        \ind\{\abs{X_{kij}-X_{k\ell j}}\le t\},
    \qquad t\ge0,
\end{equation}
and then
\begin{equation}\label{eq:Dhat}
    \qhat_{kj,\alpha}=\inf\{t\ge0:\widehat F_{kj}(t)\ge\alpha\},
    \qquad
    \hat d_{kj}=\qhat_{kj,\alpha}^2,
    \qquad
    \widehat{\mD}_k=\diag(\hat d_{k1},\ldots,\hat d_{kp}).
\end{equation}
This is a coordinatewise \(U\)-quantile scale estimator.  It is location free and does not use squared observations.

Given \(\widehat{\mD}_k\), define the full-sample standardized spatial median
\begin{equation}\label{eq:thetahat}
    \widehat{\vtheta}_k
    \in\argmin_{\vtheta\in\R^p}
    \sum_{i=1}^{n_k}\norm{\widehat{\mD}_k^{-1/2}(\vX_{ki}-\vtheta)}.
\end{equation}
The fitted within-sample signs are
\begin{equation}\label{eq:fitted-signs}
    \widehat{\vY}_{ki}=\widehat{\mD}_k^{-1/2}(\vX_{ki}-\widehat{\vtheta}_k),
    \qquad
    \widehat{\vS}_{ki}=U(\widehat{\vY}_{ki}).
\end{equation}
Our test statistic is
\begin{equation}\label{eq:Rhat}
\hat R_n^{\rm PDQ}
= -\frac1{n_1n_2}\sum_{i=1}^{n_1}\sum_{j=1}^{n_2}
 U\{\widehat{\mD}_1^{-1/2}(\vX_{1i}-\widehat{\vtheta}_2)\}\T
 U\{\widehat{\mD}_2^{-1/2}(\vX_{2j}-\widehat{\vtheta}_1)\}.
\end{equation}
Define
\begin{equation}\label{eq:hatOmegaG}
    \widehat{\mOmega}_k=n_k^{-1}\sum_{i=1}^{n_k}\widehat{\vS}_{ki}\widehat{\vS}_{ki}\T,
    \qquad
    \widehat{\mG}_k=n_k^{-1}\sum_{i=1}^{n_k}\norm{\widehat{\vY}_{ki}}^{-1}
    (\mI_p-\widehat{\vS}_{ki}\widehat{\vS}_{ki}\T).
\end{equation}

We now state the primitive assumptions and the population quantities that determine the null distribution.
Let
$N=n_1+n_2,n_0=\min(n_1,n_2)$, $x_N=\log(8pN),$ and $r_D=(x_N/n_0)^{1/2}.$

\begin{assumption}\label{ass:np} Assume that
\[
        n_1/N\to\kappa\in(0,1),
        \qquad
        x_N/n_0\to0.
\]
\end{assumption}

\begin{assumption}\label{ass:elliptical}
{\color{black}The samples are independent and identically distributed (i.i.d.), and satisfy condition \eqref{eq:model}. }%The samples are independent, each sample is i.i.d., and \eqref{eq:model} holds. 
 There are constants \(0<\underline d<\bar d<\infty\) such that
\[
        \underline d\le d^\circ_{kj}\le \bar d,
        \qquad k=1,2,
        \quad j=1,\ldots,p.
\]
\end{assumption}

\begin{assumption}\label{ass:quantile}
For the fixed \(\alpha\in(0,1)\), there exist constants \(0<q_-<q_+<\infty\), \(c_F>0\), \(C_F<\infty\), and \(\varepsilon_F>0\) such that, for \(k=1,2\),
\begin{equation}\label{eq:qbounded}
        q_-\le q_{k,\alpha}\le q_+,
\end{equation}
and, whenever \(\abs{t-q_{k,\alpha}}\le\varepsilon_F\),
\begin{equation}\label{eq:local-slope}
        c_F\abs{t-q_{k,\alpha}}
        \le
        \abs{F_{k,\Delta}(t)-\alpha}
        \le
        C_F\abs{t-q_{k,\alpha}}.
\end{equation}
\end{assumption}

\begin{assumption}\label{ass:radius}
With \(\vY_{ki}\) defined in \eqref{eq:YS}, there is a deterministic sequence \(M_{Y,n}\ge1\) such that
\begin{equation}\label{eq:inverse-radius}
    \max_{k=1,2}
    \left\{
    p^{1/2}\E\norm{\vY_{k1}}^{-1}
    +p\E\norm{\vY_{k1}}^{-2}
    +p^2\E\norm{\vY_{k1}}^{-4}
    \right\}
    \le M_{Y,n},
\end{equation}

and
\begin{equation}\label{eq:smallball}
    \max_{k=1,2}\sup_{0<t\le1}
    t^{-2}\Prb\{\norm{\vY_{k1}}\le t p^{1/2}\}
    \le M_{Y,n}.
\end{equation}

\end{assumption}

The condition in Assumption \ref{ass:radius} is a lower-tail and inverse-radius regularity condition for the standardized radius \(\norm{\vY_{k1}}\).  Appendix~\ref{app:assumption-comments} shows that it is compatible with the AR(1) and compound-symmetry benchmark shapes, relates it to the inverse-radius moment condition of \citet{ZouPengFengWang2014}, and verifies it for the normal, fixed-degree \(t\), and finite-scale mixture normal elliptical models.

Define the deterministic curvature factor
\begin{equation}\label{eq:Cn}
\mathfrak C_n
=1+M_{Y,n}^2+
\max_{k=1,2}\left[
    \left\{\norm{p^{1/2}\mG_k}_{\rm op}
    +\norm{(p^{1/2}\mG_k)^{-1}}_{\rm op}\right\}^4
    +\sup_{\norm{\va}=1}\E\{\va\T\mOmega_k^{-1/2}\vS_{ki}\}^4
\right].
\end{equation}
Set
\begin{equation}\label{eq:rates-basic}
    q_{\theta,n}=\mathfrak C_n p^{1/2}x_N/n_0,
    \qquad
    a_{S,n}=\mathfrak C_n\{r_D+n_0^{-1/2}+x_N/n_0\},
\end{equation}
and
\begin{equation}\label{eq:aK}
    a_{K,n}=\mathfrak C_n\{r_D+(x_N/n_0)^{1/2}+a_{S,n}\}.
\end{equation}

Under \(H_0\), write the common center as \(\vtheta\).  Put
\[
        \mA_{12}=\mD_1^{-1/2}\mD_2^{1/2},
        \qquad
        \mA_{21}=\mD_2^{-1/2}\mD_1^{1/2}.
\]
The matrices governing the first-order expansion are
\begin{align}
    \mM_1&=\mG_2\mA_{21}\mG_1^{-1},
    &\mK_1&=\sym(\mM_1)=\frac12(\mM_1+\mM_1\T),\label{eq:K1}\\
    \mM_2&=\mG_2^{-1}\mA_{12}\T \mG_1,
    &\mK_2&=\sym(\mM_2)=\frac12(\mM_2+\mM_2\T),\label{eq:K2}
\end{align}
and
\begin{equation}\label{eq:K3}
\mC_{12}=\mG_2^{-1}\mA_{12}\T \mG_1\mG_2\mA_{21}\mG_1^{-1},
    \qquad
    \mK_3=\mI_p+\mC_{12}\T.
\end{equation}
Let
\[
        L_{K,n}=1+\max\{\norm{\mK_1}_{\rm op},\norm{\mK_2}_{\rm op},\norm{\mK_3}_{\rm op}\}.
\]
For \(\bar{\vS}_k=n_k^{-1}\sum_i\vS_{ki}\), define
\begin{equation}\label{eq:Qn}
    Q_n
    =\bar{\vS}_1\T \mK_1\bar{\vS}_1
     +\bar{\vS}_2\T \mK_2\bar{\vS}_2
     -\bar{\vS}_1\T \mK_3\bar{\vS}_2.
\end{equation}
The full quadratic statistic has mean
\begin{equation}\label{eq:bn}
    b_n=\E Q_n
    =\frac1{n_1}\tr(\mK_1\mOmega_1)
     +\frac1{n_2}\tr(\mK_2\mOmega_2).
\end{equation}
Set
\begin{equation}\label{eq:Bfull}
    \mB_n=\begin{pmatrix}
    n_1^{-1}\mOmega_1^{1/2}\mK_1\mOmega_1^{1/2}
    &-(2\sqrt{n_1n_2})^{-1}\mOmega_1^{1/2}\mK_3\mOmega_2^{1/2}\\
    -(2\sqrt{n_1n_2})^{-1}\mOmega_2^{1/2}\mK_3\T\mOmega_1^{1/2}
    &n_2^{-1}\mOmega_2^{1/2}\mK_2\mOmega_2^{1/2}
    \end{pmatrix}.
\end{equation}
Let \(\lambda_{n1},\ldots,\lambda_{n2p}\) be the eigenvalues of \(\mB_n\), and define
\begin{equation}\label{eq:tau-def}
    \tau_n^2=2\tr(\mB_n^2)=2\sum_{r=1}^{2p}\lambda_{nr}^2.
\end{equation}

The oracle empirical diagonal correction is
\begin{equation}\label{eq:btilden}
    \tilde b_n=\frac1{n_1^2}\sum_{i=1}^{n_1}\vS_{1i}\T \mK_1\vS_{1i}
    +\frac1{n_2^2}\sum_{j=1}^{n_2}\vS_{2j}\T \mK_2\vS_{2j}.
\end{equation}
The corresponding diagonal-deleted oracle statistic is
\begin{equation}\label{eq:Un}
\begin{aligned}
    U_n=Q_n-\tilde b_n
    ={}&\frac1{n_1^2}\sum_{i\ne \ell}\vS_{1i}\T\mK_1\vS_{1\ell}
      +\frac1{n_2^2}\sum_{j\ne \ell}\vS_{2j}\T\mK_2\vS_{2\ell}  \\
    &\quad-\frac1{n_1n_2}\sum_{i=1}^{n_1}\sum_{j=1}^{n_2}\vS_{1i}\T\mK_3\vS_{2j}.
\end{aligned}
\end{equation}
The role of \(\tilde b_n\) is to remove the empirical diagonal terms; it is not used as a high-precision estimator of \(b_n\).  Define the angular row-leverage quantity
\begin{equation}\label{eq:Delta-row}
\begin{aligned}
\Delta_{{\rm row},n}
=\tau_n^{-2}\max\Bigg\{&
\frac{\norm{\mK_1\mOmega_1\mK_1}_{\rm op}}{n_1^3},
\frac{\norm{\mK_2\mOmega_2\mK_2}_{\rm op}}{n_2^3},\\
&\frac{\norm{\mK_3\mOmega_2\mK_3\T}_{\rm op}}{n_1^2n_2},
\frac{\norm{\mK_3\T\mOmega_1\mK_3}_{\rm op}}{n_1n_2^2}
\Bigg\}.
\end{aligned}
\end{equation}

\begin{assumption}\label{ass:spectral}
The matrix \(\mB_n\) satisfies \(\tr(\mB_n^2)>0\).  With
\[
        \rho_{nr}=\lambda_{nr}/\{\tr(\mB_n^2)\}^{1/2},
        \qquad r=1,\ldots,2p,
\]
there is a square-summable sequence \(\rho_1,\rho_2,\ldots\) such that, after setting \(\rho_{nr}=0\) for \(r>2p\),
\begin{equation}\label{eq:l2}
    \sum_{r\ge1}(\rho_{nr}-\rho_r)^2\to0,
    \qquad
    \rho_0^2=1-\sum_{r\ge1}\rho_r^2\ge0.
\end{equation}
The explicit approximation rates satisfy
\begin{equation}\label{eq:spectral-rates}
\begin{gathered}
    r_D\to0,
    \qquad
    q_{\theta,n}\to0,
    \qquad
    \frac{x_N}{n_0^2\tau_n}\to0,
    \qquad
    a_{K,n}\to0,\\
    \Delta_{{\rm row},n}\to0,
    \qquad
    \delta_{H,n}\to0,
    \qquad
    \delta_{{\rm dJ},n}\to0,
\end{gathered}
\end{equation}
where
\begin{align}
    \delta_{H,n}
    &=a_{K,n}^2+L_{K,n}^2a_{S,n}^2+\Delta_{{\rm row},n}+n_0^{-1},\label{eq:deltaH}\\
    \delta_{{\rm dJ},n}
    &=\mathfrak C_n\Delta_{{\rm row},n}.\label{eq:deltadj}
\end{align}
\end{assumption}

Assumption \ref{ass:spectral} is the only condition involving the shape structure; it is stated through the angular covariance matrices and does not require sparsity, bandedness, or weak correlation of \(\mR_k\).  Appendix~\ref{app:assumption-comments} gives explicit implications under bounded-spectrum, AR(1), and compound-symmetry regimes, including the relaxed dimensional requirement \(p=o\{n_0^2/\log^2(p+n_0)\}\) in the bounded-spectrum case.

\medskip\noindent We next state the main asymptotic results.

\begin{theorem}\label{thm:D}
Under {\color{black}\(H_0\) and} Assumptions \ref{ass:np}--\ref{ass:quantile},
\begin{equation}\label{eq:D-rate}
    \max_{k=1,2}\max_{1\le j\le p}
    \abs{\frac{\hat d_{kj}}{d_{kj}}-1}
    =O_p(r_D).
\end{equation}
\end{theorem}

\begin{proposition}\label{prop:derived}
Under {\color{black}\(H_0\) and} Assumptions \ref{ass:np}--\ref{ass:radius},
\begin{equation}\label{eq:Bahadur-main}
    \mD_k^{-1/2}(\widehat{\vtheta}_k-\vtheta_k)
    =\mG_k^{-1}\bar{\vS}_k+\vDelta_k,
    \qquad
    \max_{k=1,2}\norm{\vDelta_k}=O_p(q_{\theta,n}).
\end{equation}
Moreover,
\begin{equation}\label{eq:sign-pert}
    \max_{k=1,2}n_k^{-1}\sum_{i=1}^{n_k}\norm{\widehat{\vS}_{ki}-\vS_{ki}}
    =O_p(a_{S,n}),
\end{equation}
and, if \(\widehat{\mK}_1,\widehat{\mK}_2,\widehat{\mK}_3\) are formed from \(\widehat{\mD}_k,\widehat{\mG}_k\) by the analogues of \eqref{eq:K1}--\eqref{eq:K3}, then
\begin{equation}\label{eq:Khat-rate}
    \max_{r=1,2,3}\norm{\widehat{\mK}_r-\mK_r}_{\rm op}=O_p(a_{K,n}).
\end{equation}
\end{proposition}

\begin{theorem}\label{thm:oracle}
{\color{black}Under} \(H_0\), Assumptions \ref{ass:np}, \ref{ass:elliptical}, \ref{ass:radius}, and \ref{ass:spectral}.  Let
\begin{equation}\label{eq:Gamma}
    \Gamma_n=\sum_{r=1}^{2p}\lambda_{nr}(\chi_{1r}^2-1),
\end{equation}
where the \(\chi_{1r}^2\)'s are independent \(\chi_1^2\) random variables.  Then
\begin{equation}\label{eq:oracle-law}
    \sup_{t\in\R}
    \abs{
    \Prb\left(\frac{U_n}{\tau_n}\le t\right)
    -
    \Prb\left(\frac{\Gamma_n}{\tau_n}\le t\right)}\to0.
\end{equation}
Also,
\begin{equation}\label{eq:limit-mixture}
    \frac{U_n}{\tau_n}
    \dd
    \frac1{\sqrt2}\sum_{r=1}^{\infty}\rho_r(\chi_{1r}^2-1)+\rho_0Z,
\end{equation}
where \(Z\sim N(0,1)\) is independent of the chi-square variables.
\end{theorem}

Define the feasible diagonal-deleted statistic
\begin{equation}\label{eq:Tpdq}
    T_n^{\rm PDQ}=\hat R_n^{\rm PDQ}-\hat b_n,
\end{equation}
where
\begin{equation}\label{eq:bhat}
    \hat b_n
    =\frac1{n_1}\tr(\widehat{\mK}_1\widehat{\mOmega}_1)
     +\frac1{n_2}\tr(\widehat{\mK}_2\widehat{\mOmega}_2)
    =\sum_{k=1}^2\frac1{n_k^2}\sum_{i=1}^{n_k}
       \widehat{\vS}_{ki}\T\widehat{\mK}_k\widehat{\vS}_{ki}.
\end{equation}
Here \(\hat b_n\) is the empirical diagonal-deletion correction: subtracting it removes the within-sample diagonal terms in the full-sample quadratic expansion, rather than requiring \(\hat b_n\) to estimate \(b_n\) at the \(\tau_n\) scale.

\begin{theorem}\label{thm:feasible}
{\color{black}Under} \(H_0\) and Assumptions \ref{ass:np}--\ref{ass:spectral}.  Then
\begin{equation}\label{eq:main-expansion}
    \frac{T_n^{\rm PDQ}-U_n}{\tau_n}
    =O_p\left(r_D+a_{K,n}+L_{K,n}a_{S,n}+\frac{x_N}{n_0^2\tau_n}\right).
\end{equation}
Consequently,
\begin{equation}\label{eq:full-law}
    \sup_{t\in\R}
    \abs{
    \Prb\left(\frac{T_n^{\rm PDQ}}{\tau_n}\le t\right)
    -
    \Prb\left(\frac{\Gamma_n}{\tau_n}\le t\right)}\to0.
\end{equation}
\end{theorem}

\begin{corollary}\label{cor:normal}
Under the assumptions of Theorem \ref{thm:feasible}, if
\begin{equation}\label{eq:no-dom}
    \lambda_{\max}^2(\mB_n)/\tr(\mB_n^2)\to0,
\end{equation}
then
\[
        T_n^{\rm PDQ}/\tau_n\dd N(0,1).
\]
\end{corollary}

We finally introduce the Rademacher wild bootstrap to calibrate the feasible test. In the present problem, the null distribution of the centered statistic is generally a non-Gaussian weighted chi-square mixture under strong dependence, and its spectral weights depend on unknown population quantities. Although the preceding theory characterizes the limiting law, direct implementation based on that law is not convenient in practice. The bootstrap step therefore serves two purposes: it avoids explicit estimation of the unknown limiting distribution, and it remains valid both in regimes where the limit is asymptotically normal and in regimes where the leading eigenvalues are non-negligible. The Rademacher scheme is particularly convenient here because it preserves the diagonal-deleted quadratic-form structure of the statistic and yields a simple conditional approximation to the distribution of \(T_n^{\rm PDQ}=\hat R_n^{\rm PDQ}-\hat b_n\).

Let \(e_{ki}\), \(k=1,2\), \(i=1,\ldots,n_k\), be independent of the data and mutually independent Rademacher variables,
\begin{equation}\label{eq:rademacher}
    \Prstar(e_{ki}=1)=\Prstar(e_{ki}=-1)=\frac12.
\end{equation}
Then \(\Estar e_{ki}=0\), \(\Estar e_{ki}^2=1\), \(\Estar e_{ki}^3=0\), and \(\Estar e_{ki}^4=1\), where \(\Estar\) denotes expectation conditional on the data \(\calX\).  Define
\begin{equation}\label{eq:barSstar}
    \bar{\vS}_k^*=\frac1{n_k}\sum_{i=1}^{n_k}e_{ki}\widehat{\vS}_{ki},
    \qquad k=1,2,
\end{equation}
and
\begin{equation}\label{eq:Qstar}
    Q_n^*
    =\bar{\vS}_1^{*\T}\widehat{\mK}_1\bar{\vS}_1^*
     +\bar{\vS}_2^{*\T}\widehat{\mK}_2\bar{\vS}_2^*
     -\bar{\vS}_1^{*\T}\widehat{\mK}_3\bar{\vS}_2^*.
\end{equation}
Then
\begin{equation}\label{eq:Qstar-mean}
    \Estar Q_n^*=\hat b_n,
\end{equation}
and the centered bootstrap statistic is
\begin{equation}\label{eq:Tstar}
    T_n^*=Q_n^*-\hat b_n.
\end{equation}

Write \(Q_n^*=\ve\T\widehat{\mH}_n\ve\).  The symmetric \(N\times N\) matrix \(\widehat{\mH}_n=(\hat a_{ab})\) has blocks
\begin{align}\label{eq:An-blocks}
    \hat a_{1i,1\ell}&=n_1^{-2}\widehat{\vS}_{1i}\T\widehat{\mK}_1\widehat{\vS}_{1\ell},
    &\hat a_{2j,2\ell}&=n_2^{-2}\widehat{\vS}_{2j}\T\widehat{\mK}_2\widehat{\vS}_{2\ell},\\
    \hat a_{1i,2j}&=\hat a_{2j,1i}
    =-(2n_1n_2)^{-1}\widehat{\vS}_{1i}\T\widehat{\mK}_3\widehat{\vS}_{2j}.\nonumber
\end{align}
Let \(\widehat{\mH}_{n,0}\) be \(\widehat{\mH}_n\) with its diagonal set to zero.

\begin{proposition}\label{prop:boot-rates}
Under Assumptions \ref{ass:np}--\ref{ass:spectral},
\begin{align}
    \frac{\max_{1\le a\le N}\sum_{b\ne a}\hat a_{ab}^2}{\tau_n^2}
        &=O_p(\delta_{H,n}),\label{eq:boot-row}\\
    \frac{2\tr(\widehat{\mH}_{n,0}^2)}{\tau_n^2}
        &=1+O_p(\delta_{H,n}).\label{eq:boot-var}
\end{align}
Consequently, for \(\hat\tau_{*,n}^2=2\tr(\widehat{\mH}_{n,0}^2)\),
\begin{equation}\label{eq:hattau-star-rate}
    \abs{\hat\tau_{*,n}/\tau_n-1}=O_p(\delta_{H,n}^{1/2}+\delta_{H,n}).
\end{equation}
\end{proposition}

\begin{theorem}\label{thm:boot}
Under \(H_0\) and Assumptions \ref{ass:np}--\ref{ass:spectral},  then
\begin{equation}\label{eq:boot-main}
    \sup_{t\in\R}
    \abs{
    \Prstar\left(\frac{T_n^*}{\tau_n}\le t\right)
    -
    \Prb\left(\frac{T_n^{\rm PDQ}}{\tau_n}\le t\right)}
    \longrightarrow0
    \quad\hbox{in probability}.
\end{equation}
Moreover,
\begin{equation}\label{eq:boot-main-hattau}
    \sup_{t\in\R}
    \abs{
    \Prstar\left(\frac{T_n^*}{\hat\tau_{*,n}}\le t\right)
    -
    \Prb\left(\frac{T_n^{\rm PDQ}}{\hat\tau_{*,n}}\le t\right)}
    \longrightarrow0
    \quad\hbox{in probability}.
\end{equation}
If \(c_{1-\beta}^*\) is the conditional \((1-\beta)\)-quantile of \(T_n^*\), then the test
\[
        \ind\{T_n^{\rm PDQ}>c_{1-\beta}^*\}
\]
has asymptotic size \(\beta\) at continuity points of the limiting null law.
\end{theorem}

To implement the proposed procedure, one first constructs the pairwise-difference quantile estimator of the diagonal scaling matrix, then computes the corresponding spatial median and standardized spatial signs, and finally calibrates the test by the Rademacher wild bootstrap. The resulting feasible algorithm is summarized below.

\begin{algorithm}[htbp]
\caption{Feasible PDQ-based spatial-sign test with Rademacher wild bootstrap}
\label{alg:pdq-test}
\begin{algorithmic}[1]
\STATE Fix \(\alpha\), for example \(\alpha=1/4\) or \(1/2\), and compute \(\widehat{\mD}_k\) from \eqref{eq:pdq-cdf}--\eqref{eq:Dhat}.
\STATE Compute \(\widehat{\vtheta}_k\), \(\widehat{\vS}_{ki}\), \(\widehat{\mOmega}_k\), and \(\widehat{\mG}_k\) from \eqref{eq:thetahat} and \eqref{eq:hatOmegaG}.
\STATE Form \(\widehat{\mK}_1,\widehat{\mK}_2,\widehat{\mK}_3\), and compute \(\hat R_n^{\rm PDQ}\) and \(\hat b_n\) from \eqref{eq:Rhat} and \eqref{eq:bhat}.
\STATE Generate independent Rademacher multipliers according to \eqref{eq:rademacher}.
\FOR{\(b=1,\ldots,B\)}
\STATE Compute
\[
T_n^{*(b)}=Q_n^{*(b)}-\hat b_n.
\]
\ENDFOR
\STATE Reject \(H_0\) for large values of \(T_n^{\rm PDQ}=\hat R_n^{\rm PDQ}-\hat b_n\) using the empirical upper quantile of \(T_n^{*(1)},\ldots,T_n^{*(B)}\).
\end{algorithmic}
\end{algorithm}

\section{Simulation study}\label{sec:simulation}

{\color{black}In this section,} we examine the finite-sample behavior of the proposed procedure under strong dependence, where the distinction between Gaussian and non-Gaussian null calibration is most visible.  Four methods are compared throughout: PDQ, the proposed pairwise-difference-quantile spatial-sign test calibrated by the Rademacher wild bootstrap of Section~2 with $B=200$ resamples; SST, the simpler spatial-sign test of \citet{li2016}, implemented with the sign-based diagonal and location iteration of \citet{feng2016}; ART, the approximate randomization test of \citet{wangxu2022} based on the quadratic statistic {\color{black} of \cite{chenqin2010}}; and BF-F, the normal-reference $F$-type Behrens--Fisher test of \citet{zhuwangzhang2024}.

\subsection{Empirical size under the null {\color{black}hypothesis}}
We first study {\color{black}the empirical} size under the null hypothesis
\[
    H_0:\ \vtheta_1=\vtheta_2=\bm{0},
\]
with equal sample sizes $n_1=n_2=n$.  {\color{black}We consider
$    n\in\{50,100\}$ and} $
    p\in\{100,150,200,300\}.
$
{\color{black}The nominal significance level is set to} $0.05$, and all results are based on $10,000$ Monte Carlo replications.  We consider the following three elliptically symmetric models.

\smallskip
\noindent
(I) \emph{Multivariate normal distribution.}
$$
\mathbf{X}_{ki}\sim \mathcal{N}_p(\boldsymbol{\theta}_k,\mathbf{\Sigma}),
\qquad k=1,2,\ i=1,\ldots,n_k.
$$

\smallskip
\noindent
(II) \emph{Multivariate \(t\)-distribution \(t_{p,3}\).}
$$
\mathbf{X}_{ki}\text{'s are generated from } t_{p,3}
\text{ with location parameter }\boldsymbol{\theta}_k
\text{ and scatter matrix }\mathbf{\Sigma}.
$$

\smallskip
\noindent
(III) \emph{Multivariate mixture normal distribution \(\mathcal{MN}_{\gamma,9}\).}
$$
\mathbf{X}_{ki}\text{'s are generated from }
\gamma f(\boldsymbol{\theta}_k,\mathbf{\Sigma})
+
(1-\gamma)f(\boldsymbol{\theta}_k,9\mathbf{\Sigma}),
$$
which is denoted by \(\mathcal{MN}_{\gamma,9}\), where \(f(\cdot;\cdot)\) is the density
function of the \(p\)-variate multivariate normal distribution. We take \(\gamma=0.8\).

{\color{black}For the shape matrix, we consider two structures:}
\[
    \Sigma_{ab}=0.9^{|a-b|}\quad\text{for AR(1)},
    \qquad
    \Sigma_{aa}=1,\ \Sigma_{ab}=0.9\ (a\neq b)\quad\text{for CS}.
\]
Thus both designs are strongly dependent, with correlation parameter $\rho=0.9$.

Tables~\ref{tab:size-rho09-grid-normal}--\ref{tab:size-rho09-grid-mixnormal} report empirical sizes multiplied by $100$.  For each method and each dependence structure, the last row records
\[
  \mathrm{ARE}=
    \sum_{(n,p)\in\{50,100\}\times\{100,150,200,300\}}
    \bigl|100\times\widehat{\mathrm{size}}(n,p)-5\bigr|,
\]
that is, the aggregate absolute deviation from the nominal level over the eight $(n,p)$ configurations.

\begin{table}[htbp]
\centering
\caption{Empirical sizes under $H_0$ for the {\color{black}multivariate} normal distribution with $\rho=0.9$, based on $10{,}000$ Monte Carlo replications. Entries are size $\times 100$.}
\label{tab:size-rho09-grid-normal}
\begin{tabular}{ccrrrrrrrr}
\toprule
 &  & \multicolumn{4}{c}{AR(1)} & \multicolumn{4}{c}{CS} \\
\cmidrule(lr){3-6} \cmidrule(lr){7-10}
$n$ & $p$ & PDQ & SST & ART & BF-F & PDQ & SST & ART & BF-F \\
\midrule
50 & 100 & 4.70 & 6.47 & 5.46 & 5.29 & 4.41 & 6.21 & 6.10 & 5.13 \\
 & 150 & 4.62 & 6.21 & 5.50 & 5.40 & 3.95 & 5.56 & 5.53 & 4.74 \\
 & 200 & 4.72 & 6.23 & 5.74 & 5.73 & 4.21 & 5.88 & 5.62 & 4.97 \\
 & 300 & 4.14 & 5.89 & 5.48 & 5.37 & 4.13 & 5.69 & 5.76 & 4.86 \\
\midrule
100 & 100 & 5.35 & 6.65 & 5.69 & 5.61 & 4.50 & 6.02 & 5.63 & 5.07 \\
 & 150 & 5.08 & 6.42 & 5.65 & 5.44 & 4.84 & 6.39 & 5.48 & 4.79 \\
 & 200 & 5.33 & 6.61 & 5.77 & 5.61 & 4.69 & 6.32 & 5.60 & 5.06 \\
 & 300 & 5.20 & 6.45 & 5.67 & 5.55 & 4.50 & 5.79 & 5.23 & 4.53 \\
\midrule
\multicolumn{2}{c}{ARE} & 2.78 & 10.93 & 4.96 & 4.00 & 4.77 & 7.86 & 4.95 & 1.37 \\
\bottomrule
\end{tabular}
\end{table}

\begin{table}[htbp]
\centering
\caption{Empirical sizes under $H_0$ for the {\color{black}multivariate} $t$ distribution with $\rho=0.9$, based on $10{,}000$ Monte Carlo replications. Entries are size $\times 100$.}
\label{tab:size-rho09-grid-t}
\begin{tabular}{ccrrrrrrrr}
\toprule
 &  & \multicolumn{4}{c}{AR(1)} & \multicolumn{4}{c}{CS} \\
\cmidrule(lr){3-6} \cmidrule(lr){7-10}
$n$ & $p$ & PDQ & SST & ART & BF-F & PDQ & SST & ART & BF-F \\
\midrule
50 & 100 & 4.14 & 5.90 & 5.38 & 3.12 & 4.07 & 5.73 & 5.82 & 4.74 \\
 & 150 & 4.31 & 5.90 & 5.90 & 2.71 & 3.86 & 5.51 & 5.79 & 4.71 \\
 & 200 & 4.29 & 5.87 & 5.38 & 2.35 & 4.22 & 5.68 & 5.62 & 4.48 \\
 & 300 & 3.81 & 5.37 & 5.61 & 1.60 & 3.96 & 5.63 & 5.45 & 4.28 \\
\midrule
100 & 100 & 5.10 & 6.74 & 5.59 & 3.59 & 5.04 & 6.56 & 5.34 & 4.55 \\
 & 150 & 4.90 & 6.24 & 5.46 & 3.32 & 4.72 & 6.19 & 5.58 & 4.72 \\
 & 200 & 5.12 & 6.54 & 5.73 & 2.79 & 4.51 & 6.23 & 5.40 & 4.77 \\
 & 300 & 5.19 & 6.65 & 5.64 & 2.47 & 4.61 & 6.58 & 5.54 & 4.88 \\
\midrule
\multicolumn{2}{c}{ARE} & 3.96 & 9.21 & 4.69 & 18.05 & 5.09 & 8.11 & 4.54 & 2.87 \\
\bottomrule
\end{tabular}
\end{table}

\begin{table}[htbp]
\centering
\caption{Empirical sizes under $H_0$ for the {\color{black}multivariate} mixture normal distribution with $\rho=0.9$, based on $10{,}000$ Monte Carlo replications. Entries are size $\times 100$.}
\label{tab:size-rho09-grid-mixnormal}
\begin{tabular}{ccrrrrrrrr}
\toprule
 &  & \multicolumn{4}{c}{AR(1)} & \multicolumn{4}{c}{CS} \\
\cmidrule(lr){3-6} \cmidrule(lr){7-10}
$n$ & $p$ & PDQ & SST & ART & BF-F & PDQ & SST & ART & BF-F \\
\midrule
50 & 100 & 4.41 & 6.23 & 5.23 & 4.12 & 4.06 & 5.66 & 5.92 & 5.21 \\
 & 150 & 4.07 & 5.80 & 5.77 & 3.92 & 4.02 & 5.69 & 5.68 & 4.90 \\
 & 200 & 4.07 & 5.64 & 5.27 & 3.35 & 4.60 & 6.22 & 5.72 & 4.83 \\
 & 300 & 3.92 & 5.38 & 5.66 & 2.86 & 4.42 & 5.97 & 5.82 & 5.06 \\
\midrule
100 & 100 & 4.79 & 6.45 & 5.69 & 4.67 & 4.92 & 6.46 & 5.47 & 4.78 \\
 & 150 & 4.81 & 6.49 & 5.51 & 4.37 & 4.58 & 6.35 & 5.89 & 4.98 \\
 & 200 & 5.07 & 6.39 & 5.48 & 4.44 & 4.92 & 6.40 & 6.02 & 5.31 \\
 & 300 & 4.91 & 6.14 & 5.62 & 3.95 & 4.70 & 6.45 & 5.62 & 5.03 \\
\midrule
\multicolumn{2}{c}{ARE} & 4.09 & 8.52 & 4.23 & 8.32 & 3.78 & 9.20 & 6.14 & 1.12 \\
\bottomrule
\end{tabular}
\end{table}

The size results lead to several clear conclusions. First, the proposed PDQ test provides the most accurate and stable control of the nominal level across the entire simulation grid. Under the AR(1) design, its empirical rejection probabilities remain uniformly close to the target level under all three distributional settings, including the heavy-tailed \(t_3\) model and the mixture normal model. This uniform behavior is particularly important in the present problem, because the competing methods are much more sensitive to departures from Gaussianity and to strong dependence. By contrast, the PDQ procedure maintains reliable size control throughout, which confirms the practical effectiveness of the new diagonal normalization and the Rademacher wild bootstrap calibration.

Second, SST is systematically liberal in nearly all settings considered here. This pattern is fully consistent with the theoretical discussion: its normal calibration is justified only under substantially stronger structural assumptions, and these assumptions are violated in the strongly dependent designs studied in this section. As a result, the apparent advantage of SST in some of the power plots must be interpreted with caution, since part of that gain is explained by over-rejection under the null rather than by genuinely superior power.

Third, under the CS design, BF-F remains reasonably competitive for the Gaussian and mixture normal models, but its behavior deteriorates substantially under the heavy-tailed \(t_3\) model, where it becomes distinctly conservative. In contrast, PDQ continues to track the nominal level closely even in this difficult setting. Taken together, these findings show that the proposed method offers the best overall size performance among the procedures under comparison. In particular, when dependence is strong and the underlying distribution is heavy-tailed or contaminated, PDQ is the only method that consistently delivers both accurate calibration and stable behavior over the full range of sample sizes and dimensions considered here.

\subsection{Power under dense alternatives}

{\color{black}Next, we evaluate the power performance under the dense alternative
\[
    \vtheta_1=\bm{0},
    \qquad
    \vtheta_2=\delta\bm{1}_p/\sqrt{p},
\]
with $n_1=n_2=100$, $p=200$, and $\rho=0.9$. as in the size study, we use the same three data generating distributions, with empirical power calculated from 1,000 Monte Carlo replications. For the AR(1) structure, we specify
$\delta \in \{0.25, 0.5, 0.75, 1, 1.25, 1.5, 1.75, 2, 2.25, 2.5\},$
while for the CS structure, we adopt the strong grid
$\delta \in \{1, 2, 3, 4, 5, 6, 7, 8, 9, 10\}.$}

\begin{figure}[htbp]
\centering
\includegraphics[width=.96\textwidth]{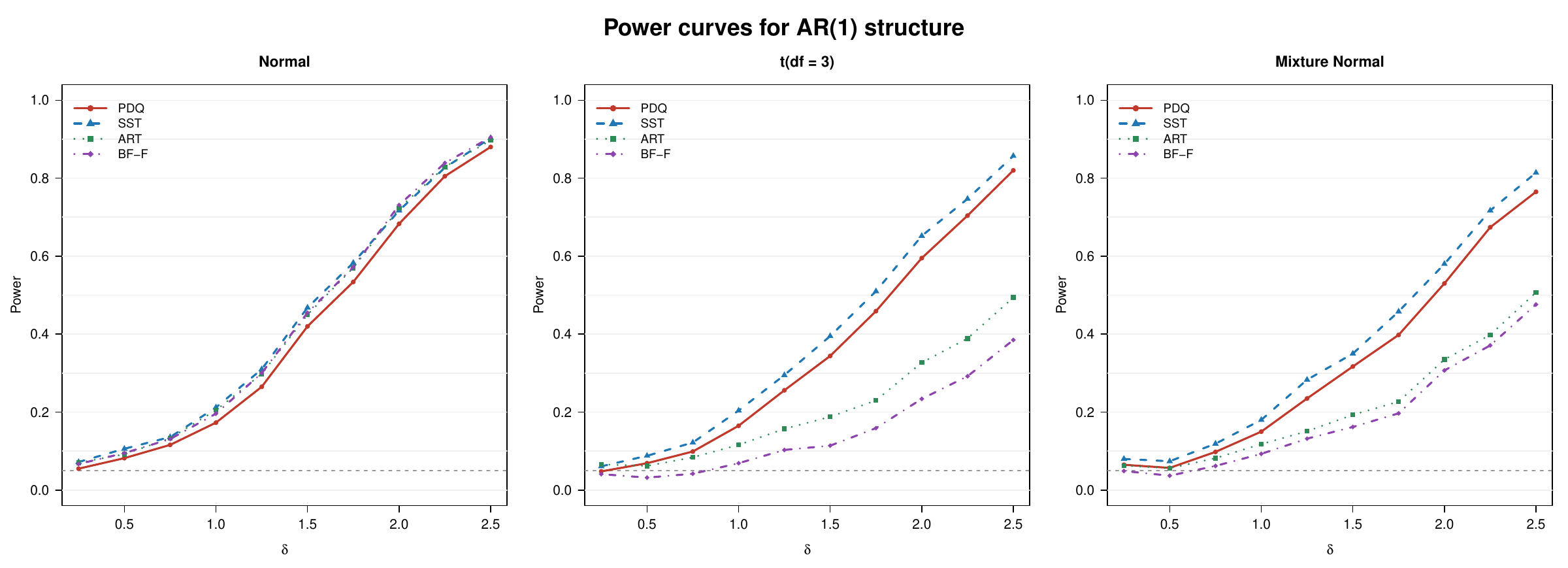}
\caption{Power curves under the AR(1) structure with $\rho=0.9$, $n_1=n_2=100$, $p=200$. {\color{black}The dense shift is $\vtheta_2 = \delta\bm{1}_p/\sqrt{p}$, with results based on 1,000 Monte Carlo replications per signal level.}}
\label{fig:power-ar1-rho09}
\end{figure}

\begin{figure}[htbp]
\centering
\includegraphics[width=.96\textwidth]{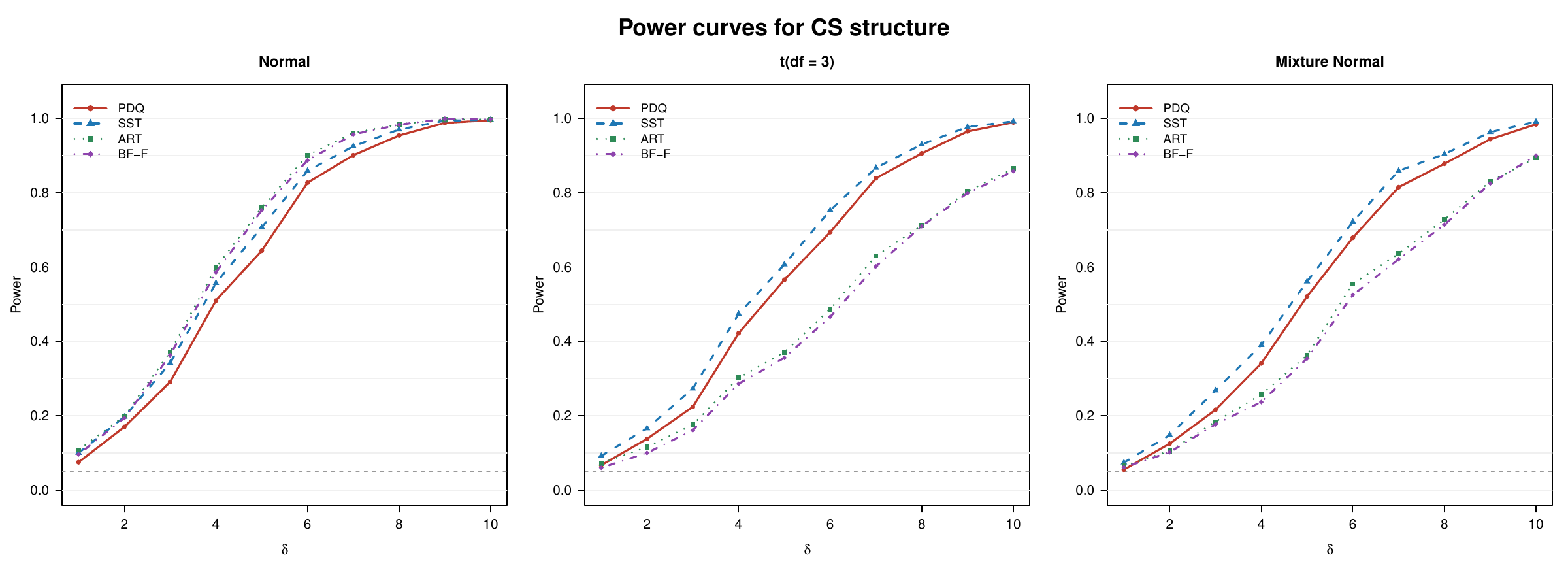}
\caption{Power curves under the CS structure with $\rho=0.9$, $n_1=n_2=100$, $p=200$. {\color{black}The dense shift is $\vtheta_2 = \delta\bm{1}_p/\sqrt{p}$, with results based on 1,000 Monte Carlo replications per signal level.}}
\label{fig:power-cs-rho09}
\end{figure}

The power curves in Figures~\ref{fig:power-ar1-rho09}--\ref{fig:power-cs-rho09} further illustrate the finite-sample advantages of the proposed PDQ procedure. {\color{black}Under the Gaussian model, the power of PDQ is marginally exceeded by ART and BF-F; nonetheless, it tracks their performance closely and remains of the same order across all considered signal strengths.} The advantage of PDQ becomes more pronounced under non-Gaussian designs. In particular, under the \(t_3\) model and the mixture normal model, PDQ delivers uniformly higher power than ART and BF-F over the practically relevant signal range under both AR(1) and compound-symmetry dependence. This pattern is especially clear when the signal is moderate, where robustness against heavy tails and contamination is most consequential for finite-sample performance. These results are consistent with the theoretical motivation of the proposed method: by combining spatial signs with the PDQ diagonal scaling, the test remains stable under strong dependence while being substantially less sensitive to tail heaviness and scale contamination.

The SST curves are in several cases above those of PDQ, and sometimes substantially so. However, these larger rejection probabilities should be interpreted with caution. As shown in Tables~\ref{tab:size-rho09-grid-normal}--\ref{tab:size-rho09-grid-mixnormal}, SST exhibits noticeable size distortion in the same settings, particularly under strong dependence and non-Gaussian models. Hence part of its apparent power gain is attributable to over-rejection under the null rather than to genuinely better discrimination under the alternative. When size control is taken into account jointly with power, PDQ provides the most reliable overall performance among the methods considered, especially in the heavy-tailed and contaminated settings for which robustness is essential.

\section{Conclusion}

This paper studies two-sample spatial-sign inference under high-dimensional elliptical models with arbitrary dependence structures. The main difficulty in this setting is that, under strong dependence, the classical spatial-sign statistics of \citet{feng2016} and \citet{li2016} are no longer asymptotically normal, and their diagonal scaling step is no longer valid. To address these issues, we propose a new pairwise-difference quantile estimator for the diagonal scaling matrix and develop a feasible spatial-sign test that remains valid without sparsity assumptions on the scatter matrix. We show that the feasible statistic admits a weighted chi-square mixture limit under the null, and we further establish the consistency of a Rademacher wild bootstrap for practical calibration. The numerical results indicate that the proposed procedure achieves reliable size control and competitive power, especially under heavy-tailed and contaminated distributions.

Several problems merit further investigation. One direction is to extend the present theory to more general multi-sample or MANOVA-type problems under elliptical symmetry. Another direction is to develop adaptive procedures that combine the present robust spatial-sign statistic with other complementary tests in order to improve power across a broader range of signal configurations.

\appendix

\section{Proofs of the main results}\label{app:proofs}

\subsection{Auxiliary lemmas}
\label{app:aux}

\begin{lemma}\label{lem:ustat-tail}
Let
\[
    U_n=\binom n2^{-1}\sum_{1\le i<j\le n}h(X_i,X_j),
    \qquad
    \theta=\E h(X_1,X_2),
\]
where \(X_i\)'s are independent and \(0\le h(\cdot)\le1\). Then, for every \(x>0\),
\[
    \Prb\{\abs{U_n-\theta}>x\}
    \le 2\exp(-c n x^2).
\]
Hence, for deterministic points \(t_{kj}^{\sigma}\), \(k=1,2\), \(j\le p\), \(\sigma\in\{-,+\}\),
\[
\Prb\left\{
\max_{k,j,\sigma}
\abs{\widehat F_{kj}(t_{kj}^{\sigma})-F_{kj}(t_{kj}^{\sigma})}>x
\right\}
\le 4p\exp(-c n_0x^2).
\]
\end{lemma}

\begin{proof}
Let \(m=\lfloor n/2\rfloor\). The symmetrization identity is
\[
    U_n=\frac1{n!}\sum_{\pi}\frac1m
    \sum_{r=1}^m h(X_{\pi(2r-1)},X_{\pi(2r)}).
\]
For every permutation \(\pi\), the summands in the inner average are independent and lie in \([0,1]\). Hence
\[
\begin{aligned}
\Prb\{U_n-\theta>x\}
&\le \frac1{n!}\sum_{\pi}
\Prb\left\{m^{-1}\sum_{r=1}^m
[h(X_{\pi(2r-1)},X_{\pi(2r)})-\theta]>x\right\} \\
&\le \exp(-2mx^2).
\end{aligned}
\]
The lower tail is identical, and \(m\ge n/3\) for \(n\ge2\).
\end{proof}

\begin{lemma}\label{lem:scalar-bern-general}
Let \(W_1,\ldots,W_n\) be independent mean-zero random variables with
\[
    |W_i|\le L,
    \qquad
    v=\sum_{i=1}^n\E W_i^2 .
\]
Then, for every \(s>0\),
\[
    \Prb\left\{\sum_{i=1}^nW_i>s\right\}
    \le
    \exp\left\{-\frac{s^2}{2v+2Ls/3}\right\}.
\]
\end{lemma}
\begin{proof}
For \(0<\lambda<3/L\),
\[
    \log\E e^{\lambda W_i}
    \le \frac{\lambda^2\E W_i^2}{2(1-\lambda L/3)} .
\]
Hence
\[
    \Prb\left\{\sum_iW_i>s\right\}
    \le
    \exp\left\{-\lambda s+
    \frac{\lambda^2v}{2(1-\lambda L/3)}\right\}.
\]
The choice
\[
    \lambda=\frac{s}{v+Ls/3}
\]
gives the stated exponent.
\end{proof}

\begin{lemma}\label{lem:bern-ind}
Let \(Z_1,\ldots,Z_n\) be independent random variables with \(0\le Z_i\le1\),
\[
    \bar Z=n^{-1}\sum_{i=1}^nZ_i,
    \qquad
    \bar\mu=n^{-1}\sum_{i=1}^n\E Z_i .
\]
For every \(x>0\),
\[
    \Prb\left\{
    \bar Z-\bar\mu
    >2\left(\frac{\bar\mu x}{n}\right)^{1/2}+\frac{2x}{3n}
    \right\}
    \le e^{-x}.
\]
Consequently, if \(Z_i(r)=\ind\{\norm{\vY_i}\le r\}\) and
\[
    \Prb\{\norm{\vY_i}\le r\}\le M_Yr^2/p,
    \qquad 0<r\le p^{1/2},
\]
then
\[
    n^{-1}\sum_{i=1}^nZ_i(r)
    =O_p\left(M_Yr^2/p+x/n\right)
\]
for every deterministic \(x\ge1\).
\end{lemma}

\begin{proof}
Apply Lemma \ref{lem:scalar-bern-general} with
\[
    W_i=Z_i-\E Z_i,
    \qquad
    L=1,
    \qquad
    v=\sum_{i=1}^n\Var(Z_i)
       \le \sum_{i=1}^n\E Z_i=n\bar\mu .
\]
Then
\[
    \Prb\left\{\sum_{i=1}^n(Z_i-\E Z_i)>s\right\}
    \le
    \exp\left\{-\frac{s^2}{2n\bar\mu+2s/3}\right\}.
\]
Set
\[
    s=2(n\bar\mu x)^{1/2}+2x/3.
\]
Then
\[
    s^2\ge 4n\bar\mu x+4x^2/9,
    \qquad
    2n\bar\mu+2s/3\le 4n\bar\mu+4x/3,
\]
so the exponent is bounded below by \(cx\). If \(\bar\mu(r)\le M_Yr^2/p\), then
\[
    \left\{\frac{\bar\mu(r)x}{n}\right\}^{1/2}
    \le \frac12\bar\mu(r)+\frac{x}{2n}
    \le C\{M_Yr^2/p+x/n\}.
\]
\end{proof}

\begin{lemma}\label{lem:matrix-bern-general}
Let \(\mathbf X_1,\ldots,\mathbf X_n\) be independent self-adjoint random matrices with
\[
    \E\mathbf X_i=\mathbf 0,
    \qquad
    \norm{\mathbf X_i}_{\rm op}\le L .
\]
Set
\[
    v=\left\|\sum_{i=1}^n \E\mathbf X_i^2\right\|_{\rm op}.
\]
Then, for every \(t>0\),
\[
    \Prb\left\{\left\|\sum_{i=1}^n\mathbf X_i\right\|_{\rm op}>t\right\}
    \le
    2p\exp\left\{-\frac{t^2/2}{v+Lt/3}\right\}.
\]
Consequently, for every \(x\ge\log(2p)\),
\[
    \Prb\left\{\left\|\sum_{i=1}^n\mathbf X_i\right\|_{\rm op}
    >C\left(\sqrt{vx}+Lx\right)\right\}
    \le 2p e^{-x}.
\]
\end{lemma}
\begin{proof}
{\color{black}According to \citep{tropp2012}, it holds that}
\[
    \Prb\left\{\lambda_{\max}\left(\sum_i\mathbf X_i\right)>t\right\}
    \le p\exp\left\{-\frac{t^2/2}{v+Lt/3}\right\},
\]
and the same inequality for \(\lambda_{\max}(-\sum_i\mathbf X_i)\).  Therefore
\[
    \Prb\left\{\left\|\sum_i\mathbf X_i\right\|_{\rm op}>t\right\}
    \le 2p\exp\left\{-\frac{t^2/2}{v+Lt/3}\right\}.
\]
For \(t=C(\sqrt{vx}+Lx)\),
\[
    \frac{t^2/2}{v+Lt/3}\ge x
\]
for a sufficiently large universal \(C\).
\end{proof}

\begin{lemma}\label{lem:matrix-bern-L}
Let \(\vY_1,\ldots,\vY_n\) be independent copies of \(\vY\in\R^p\). Define
\[
    \mL(\vY)=\norm{\vY}^{-1}\{\mI_p-U(\vY)U(\vY)\T\},
    \qquad
    \mG=\E\mL(\vY).
\]
Assume
\[
    p\E\norm{\vY}^{-2}\le M_Y,
    \qquad
    \sup_{0<t\le1}t^{-2}\Prb\{\norm{\vY}\le t p^{1/2}\}\le M_Y .
\]
For \(x\ge\log(2p)\), put
\[
    \ell=(px/n)^{1/2},
    \qquad
    \mL_i^{>} = \mL(\vY_i)\ind\{\norm{\vY_i}>\ell\},
    \qquad
    \mL_i^{<} = \mL(\vY_i)\ind\{\norm{\vY_i}\le\ell\}.
\]
Then
\[
    \left\|n^{-1}\sum_{i=1}^n\{\mL(\vY_i)-\mG\}\right\|_{\rm op}
    =O_p\left(M_Y\sqrt{\frac{x}{np}}\right).
\]
\end{lemma}

\begin{proof}
For
\[
    \mX_i=\mL_i^{>}-\E\mL_i^{>},
\]
\[
    \norm{\mX_i}_{\rm op}\le2\ell^{-1},
    \qquad
    \left\|\sum_{i=1}^n\E \mX_i^2\right\|_{\rm op}
    \le n\E\norm{\mL(\vY)}_{\rm op}^2
    \le nM_Y/p.
\]
Lemma \ref{lem:matrix-bern-general}, with
\[
    v\le nM_Y/p,
    \qquad
    L=2\ell^{-1},
\]
gives
\[
\Prb\left[
\left\|n^{-1}\sum_i\mX_i\right\|_{\rm op}
>C\left\{\sqrt{\frac{M_Yx}{np}}+\frac{x}{n\ell}\right\}
\right]
\le2p e^{-x},
\qquad
\frac{x}{n\ell}=\sqrt{\frac{x}{np}}.
\]
For the lower-tail part,
\[
    \norm{\E \mL_i^{<}}_{\rm op}
    \le \E\{\norm{\vY}^{-1}\ind(\norm{\vY}\le\ell)\}
    \le (\E\norm{\vY}^{-2})^{1/2}\Prb(\norm{\vY}\le\ell)^{1/2}
    \le M_Y\sqrt{\frac{x}{np}}.
\]
Moreover,
\[
\begin{aligned}
\left\|n^{-1}\sum_i\mL_i^{<}\right\|_{\rm op}
&\le
\left(n^{-1}\sum_i\norm{\vY_i}^{-2}\right)^{1/2}
\left(n^{-1}\sum_i\ind(\norm{\vY_i}\le\ell)\right)^{1/2}  \\
&=O_p\left((M_Y/p)^{1/2}\right)
  O_p\left(M_Y\ell^2/p+x/n\right)^{1/2} \\
&=O_p\left(M_Y\sqrt{\frac{x}{np}}\right).
\end{aligned}
\]
\end{proof}

\begin{lemma}\label{lem:dejong-tool}
Let \(\bm{V}_{n1},\ldots,\bm{V}_{nN}\) be independent mean-zero random vectors and let
\[
    H_n=\sum_{a\ne b}\bm{V}_{na}\T \mathbf{C}_{nab}\bm{V}_{nb},
    \qquad
    \sigma_n^2=\Var(H_n)>0.
\]
Define the row influence and fourth-projection ratios by
\[
    h_n=\sigma_n^{-2}\max_a
    \sum_{b\ne a}\E\{\bm{V}_{na}\T \mathbf{C}_{nab}\bm{V}_{nb}\}^2,
\]
\[
    d_n=\sigma_n^{-4}
    \sum_a\E\left[
    \left\{\sum_{b\ne a}\E(\bm{V}_{na}\T \mathbf{C}_{nab}\bm{V}_{nb}\mid \bm{V}_{na})^2\right\}^2
    \right].
\]
If \(\bm{G}_{na}\) are independent Gaussian vectors with the same covariance matrices as \(\bm{V}_{na}\), then there is a deterministic function \(\omega_{\rm dJ}\), \(\omega_{\rm dJ}(u)\downarrow0\), such that
\[
    \sup_t\left|
    \Prb(H_n/\sigma_n\le t)-\Prb(H_{n,G}/\sigma_n\le t)
    \right|
    \le \omega_{\rm dJ}(h_n+d_n),
\]
where \(H_{n,G}=\sum_{a\ne b}\bm{G}_{na}\T \mathbf{C}_{nab}\bm{G}_{nb}\). The same bound holds conditionally on the data.
\end{lemma}

\begin{proof}
This is the coefficient form of de Jong's comparison theorem for generalized degenerate quadratic forms \citep{dejong1987}.
\end{proof}

\begin{lemma}\label{lem:HW}
For symmetric matrices \(\mathbf A,\mathbf B\in\R^{m\times m}\), let \(\lambda_1(\mathbf A),\ldots,\lambda_m(\mathbf A)\) and \(\lambda_1(\mathbf B),\ldots,\lambda_m(\mathbf B)\) be arranged in nonincreasing order. Then
\[
    \sum_{r=1}^m\{\lambda_r(\mathbf A)-\lambda_r(\mathbf B)\}^2
    \le \norm{\mathbf A-\mathbf B}_F^2.
\]
\end{lemma}
\begin{proof}
This is the Hoffman--Wielandt inequality for normal matrices \citep{hoffmanwielandt1953}, restricted to symmetric matrices.
\end{proof}

\begin{lemma}\label{lem:be-chisq}
Let \(\zeta_r=\chi_{1r}^2-1\) be independent and
\[
    W_n=\sum_r a_{nr}\zeta_r,
    \qquad
    s_n^2=2\sum_r a_{nr}^2.
\]
If \(s_n^2>0\), then
\[
    \sup_t\abs{\Prb(W_n/s_n\le t)-\Phi(t)}
    \le C\frac{\sum_r\abs{a_{nr}}^3}{(\sum_r a_{nr}^2)^{3/2}}
    \le C\frac{\max_r\abs{a_{nr}}}{(\sum_r a_{nr}^2)^{1/2}}.
\]
\end{lemma}

\begin{proof}
Set \(X_{nr}=a_{nr}\zeta_r/s_n\). Then
\[
    \sum_r\E X_{nr}^2=1,
    \qquad
    \sum_r\E\abs{X_{nr}}^3
    =s_n^{-3}\E\abs{\chi_1^2-1}^3\sum_r\abs{a_{nr}}^3.
\]
The Berry--Esseen inequality \citep[see][]{petrov1995} gives
\[
    \sup_t\abs{\Prb(W_n/s_n\le t)-\Phi(t)}
    \le C\sum_r\E\abs{X_{nr}}^3.
\]
Since
\[
    \sum_r\abs{a_{nr}}^3
    \le \max_r\abs{a_{nr}}\sum_r a_{nr}^2,
\]
the second bound follows.
\end{proof}

\subsection{Pairwise-quantile diagonal standardization}
\label{app:D}

\begin{lemma}\label{lem:marginal-id}
Under \eqref{eq:model} {\color{black}and $H_0$}, for every \(k,j\),
\[
    \Prb\{\abs{X_{ki j}-X_{k\ell j}}\le t\}
    =F_{k,\Delta}(t/(d^\circ_{kj})^{1/2}),
    \qquad t\ge0.
\]
Thus the \(\alpha\)-quantile of \(\abs{X_{ki j}-X_{k\ell j}}\) is \(d_{kj}^{1/2}\), with \(d_{kj}\) defined in \eqref{eq:Dtarget}.
\end{lemma}

\begin{proof}
Let \(a_{kj}\T\) be the \(j\)th row of \(\mR_k^{1/2}\).  Since \(R_{k,jj}=1\), \(\norm{a_{kj}}=1\).  Hence
\[
    X_{ki j}-\theta_{kj}
    =(d^\circ_{kj})^{1/2}\xi_{ki}a_{kj}\T \vu_{ki},
\]
\[
    \xi_{ki}a_{kj}\T \vu_{ki}\stackrel{d}{=}\xi_{ki}e_1\T \vu_{ki},
\]
\[
    \frac{\abs{X_{ki j}-X_{k\ell j}}}{(d^\circ_{kj})^{1/2}}
    \stackrel{d}{=}\abs{Z_{k1}-Z_{k2}}.
\]
Therefore
\[
    Q_\alpha\{\abs{X_{ki j}-X_{k\ell j}}\}
    =q_{k,\alpha}(d^\circ_{kj})^{1/2}=d_{kj}^{1/2}.
\]
\end{proof}

\begin{lemma}\label{lem:uq}
Let \(\hat d_{kj}\) be defined in \eqref{eq:Dhat}.  Under {\color{black} $H_0$ and} Assumptions \ref{ass:np}--\ref{ass:quantile}, for every sufficiently large constant \(M\),
\[
\Prb\left\{\max_{k,j}\abs{\frac{\qhat_{kj,\alpha}}{d_{kj}^{1/2}}-1}>Mr_D\right\}
\le C\exp(-cM^2x_N+\log p).
\]
Consequently,
\[
    \max_{k,j}\abs{\frac{\qhat_{kj,\alpha}}{d_{kj}^{1/2}}-1}=O_p(r_D).
\]
\end{lemma}

\begin{proof}
The desired event is obtained from the two inequalities below.  For \(\delta=Mr_D\) and \(M r_D\le\varepsilon_F/q_+\), put
\[
    t_{kj}^+=d_{kj}^{1/2}(1+\delta),
    \qquad
    t_{kj}^-=d_{kj}^{1/2}(1-\delta).
\]
By Lemma \ref{lem:marginal-id} and \eqref{eq:local-slope},
\[
    F_{kj}(t_{kj}^+)-\alpha
    =F_{k,\Delta}\{q_{k,\alpha}(1+\delta)\}-\alpha
    \ge c\delta,
\]
\[
    \alpha-F_{kj}(t_{kj}^-)
    =\alpha-F_{k,\Delta}\{q_{k,\alpha}(1-\delta)\}
    \ge c\delta.
\]
For
\[
    h_t(x,y)=\ind\{\abs{x_j-y_j}\le t\},
    \qquad 0\le h_t\le1,
\]
Lemma \ref{lem:ustat-tail} gives
\[
\begin{aligned}
\Prb\{\widehat F_{kj}(t_{kj}^+)<\alpha\}
&\le \Prb\{\widehat F_{kj}(t_{kj}^+)-F_{kj}(t_{kj}^+)<-c\delta\} \\
&\le 2\exp(-c n_k\delta^2),
\end{aligned}
\]
\[
\begin{aligned}
\Prb\{\widehat F_{kj}(t_{kj}^-)\ge\alpha\}
&\le \Prb\{\widehat F_{kj}(t_{kj}^-)-F_{kj}(t_{kj}^-)>c\delta\} \\
&\le 2\exp(-c n_k\delta^2).
\end{aligned}
\]
Union over \(k,j\):
\[
\Prb\left\{\max_{k,j}\abs{\frac{\qhat_{kj,\alpha}}{d_{kj}^{1/2}}-1}>\delta\right\}
\le 4p\exp(-cn_0\delta^2)
=4p\exp(-cM^2x_N).
\]
\end{proof}

\begin{proof}[Proof of Theorem \ref{thm:D}]
By Lemma \ref{lem:uq},
\[
    \max_{k,j}\abs{\qhat_{kj,\alpha}/d_{kj}^{1/2}-1}=O_p(r_D).
\]
Since \(r_D\to0\),
\[
    \frac{\hat d_{kj}}{d_{kj}}-1
    =\left(\frac{\qhat_{kj,\alpha}}{d_{kj}^{1/2}}-1\right)
     \left(\frac{\qhat_{kj,\alpha}}{d_{kj}^{1/2}}+1\right),
\]
\[
    \max_{k,j}\abs{\hat d_{kj}/d_{kj}-1}=O_p(r_D).
\]
\end{proof}

\subsection{Spatial median and plug-in matrices}
\label{app:median}

\begin{lemma}\label{lem:sign-calculus}
For \(\vx\ne0\), define
\[
    \mL(\vx)=\norm{\vx}^{-1}\{\mI_p-U(\vx)U(\vx)\T\}.
\]
Then, for all \(\vh\),
\[
    \norm{U(\vx+\vh)-U(\vx)}\le2\min\{1,\norm{\vh}/\norm{\vx}\},
\]

and if \(\norm{\vh}\le\norm{\vx}/2\),
\[
    \norm{U(\vx+\vh)-U(\vx)-\mL(\vx)\vh}\le C\norm{\vh}^2\norm{\vx}^{-2}.
\]
Consequently,
\[
    \E\norm{U(\vY+\vh)-U(\vY)-\mL(\vY)\vh}
    \le C\norm{\vh}^2\E\norm{\vY}^{-2}+2\Prb(\norm{\vY}\le2\norm{\vh}).
\]
\end{lemma}

\begin{proof}
For \(f(t)=U(\vx+t\vh)\),
\[
    f'(t)=\norm{\vx+t\vh}^{-1}\{\mI_p-U(\vx+t\vh)U(\vx+t\vh)\T\}\vh,
\]
\[
    \norm{f'(t)}\le \norm{\vh}\norm{\vx+t\vh}^{-1},
    \qquad
    \norm{f''(t)}\le C\norm{\vh}^2\norm{\vx}^{-2}
\]
whenever \(\norm{\vh}\le\norm{\vx}/2\).  Thus
\[
    U(\vx+\vh)-U(\vx)=\int_0^1f'(t)dt,
\]
\[
    U(\vx+\vh)-U(\vx)-\mL(\vx)\vh=\int_0^1\{f'(t)-f'(0)\}dt.
\]
The last display follows by splitting \(\{\norm{\vY}>2\norm{\vh}\}\) and its complement.
\end{proof}

Let
\[
    \mathcal E_{D,n}(M)=
    \left\{\max_{k,j}\abs{\hat d_{kj}/d_{kj}-1}\le Mr_D\right\}.
          \]
Theorem \ref{thm:D} gives, for each large fixed \(M\),
\[
    \Prb\{\mathcal E_{D,n}(M)^c\}\le C_M\exp(-c_Mx_N+\log p).
\]
Define
\[
    \widehat{\mA}_{12}=\widehat{\mD}_1^{-1/2}\widehat{\mD}_2^{1/2},
    \qquad
    \widehat{\mA}_{21}=\widehat{\mD}_2^{-1/2}\widehat{\mD}_1^{1/2}.
\]
On \(\mathcal E_{D,n}(M)\),
\begin{equation}\label{eq:D-event}
    \max_k\norm{\widehat{\mD}_k^{-1/2}\mD_k^{1/2}-\mI_p}_{\rm op}
    +\norm{\hat \mA_{12}-\mA_{12}}_{\rm op}
    +\norm{\hat \mA_{21}-\mA_{21}}_{\rm op}
    \le CMr_D.
\end{equation}

\begin{lemma}\label{lem:median-rate}
Under Assumptions \ref{ass:np}--\ref{ass:radius},
\[
    \mD_k^{-1/2}(\widehat{\vtheta}_k-\vtheta_k)
    =\mG_k^{-1}\bar{\vS}_k+\vDelta_k,
    \qquad
    \max_k\norm{\vDelta_k}=O_p(q_{\theta,n}).
\]
\end{lemma}

\begin{proof}
Put
\[
    \widehat{\mE}_k=\widehat{\mD}_k^{-1/2}\mD_k^{1/2}-\mI_p,
    \qquad
    \vdelta_k=\mD_k^{-1/2}(\widehat{\vtheta}_k-\vtheta_k),
\]
\[
    \Phi_{k,n}(\vt,\mE)
    =n_k^{-1}\sum_{i=1}^{n_k}U\{(\mI_p+\mE)(\vY_{ki}-\vt)\},
\]
\[
    \widehat{\mG}_{k,n}(\mE)
    =n_k^{-1}\sum_{i=1}^{n_k}
      \mL\{(\mI_p+\mE)\vY_{ki}\}(\mI_p+\mE).
\]
Since
\[
    \widehat{\mD}_k^{-1/2}(\vX_{ki}-\widehat{\vtheta}_k)
    = (\mI_p+\widehat{\mE}_k)(\vY_{ki}-\vdelta_k),
\]
the first-order condition for \eqref{eq:thetahat} is
\[
    \Phi_{k,n}(\vdelta_k,\widehat{\mE}_k)=\bm 0.
\]
On \(\mathcal E_{D,n}(M)\),
\[
    \norm{\widehat{\mE}_k}_{\rm op}\le CMr_D,
    \qquad
    \Prb\{\mathcal E_{D,n}(M)^c\}\le C_M\exp(-c_Mx_N+\log p).
\]
For every deterministic diagonal \(\mE\) satisfying \(\norm{\mE}_{\rm op}\le1/2\), central symmetry gives
\[
    \E U\{(\mI_p+\mE)\vY_{k1}\}=\bm0,
    \qquad
    \E U(\vY_{k1})=\bm0.
\]
Moreover,
\[
\begin{aligned}
&\norm{U\{(\mI_p+\mE)\vY\}-U(\vY)}
        \le 2\norm{\mE}_{\rm op},\qquad
\E\norm{U\{(\mI_p+\mE)\vY\}-U(\vY)}^2
        \le C\norm{\mE}_{\rm op}^2,
\end{aligned}
\]
\[
    \E\left\|n_k^{-1}\sum_{i=1}^{n_k}
   \left[U\{(\mI_p+\mE)\vY_{ki}\}-U(\vY_{ki})\right]\right\|^2
        \le C\norm{\mE}_{\rm op}^2/n_k .
\]
The pairwise-quantile Bahadur expansion gives, uniformly in \(j\),
\[
    \widehat e_{kj}:=(\hat d_{kj}/d_{kj})^{-1/2}-1
    =n_k^{-1}\sum_{\ell=1}^{n_k}\varphi_{kj}(Y_{k\ell,j})+\rho_{kj},
\]
\[
    \E\varphi_{kj}(Y_{k1,j})=0,
    \qquad
    \varphi_{kj}(-y)=\varphi_{kj}(y),
    \qquad
    \max_j\abs{\rho_{kj}}=O_p(x_N/n_k),
\]
\[
    \max_j\E\varphi_{kj}^2(Y_{k1,j})\le C,
    \qquad
    \max_j\abs{\widehat e_{kj}}=O_p(r_D).
\]
Let
\[
    \mH_{ki}=\mL(\vY_{ki})\diag(\vY_{ki}),
    \qquad
    \widehat{\ve}_k=(\widehat e_{k1},\ldots,\widehat e_{kp})\T.
\]
Then \(\norm{\mH_{ki}\va}\le\norm{\va}_{\infty}\),
\(\E\mH_{ki}=\bm0\), and
\[
    \E\{\mH_{ki}\varphi_k(\vY_{ki})\}=\bm0,
    \qquad
    \varphi_k(\vY_{ki})=(\varphi_{k1}(Y_{ki,1}),\ldots,
    \varphi_{kp}(Y_{ki,p}))\T,
\]
because \(\mH_{ki}\) is odd and \(\varphi_k\) is even under
\(\vY_{ki}\mapsto-\vY_{ki}\). Consequently,
\[
\begin{aligned}
&\E\left\|
      n_k^{-1}\sum_{i=1}^{n_k}\mH_{ki}\widehat{\ve}_k
      \right\|^2\cr
&\quad\le
 C\E\left\|
      n_k^{-2}\sum_{i=1}^{n_k}\sum_{\ell=1}^{n_k}
      \mH_{ki}\varphi_k(\vY_{k\ell})
      \right\|^2
   +C\E\left\|
      n_k^{-1}\sum_{i=1}^{n_k}\mH_{ki}\bm\rho_k
      \right\|^2\cr
&\quad\le Cn_k^{-2}+C x_N^2 n_k^{-2},
       \qquad
       \bm\rho_k=(\rho_{k1},\ldots,\rho_{kp})\T.
\end{aligned}
\]
Together with Lemma \ref{lem:sign-calculus}, \eqref{eq:inverse-radius} and
\eqref{eq:smallball}, this yields
\[
    \norm{\Phi_{k,n}(\bm0,\widehat{\mE}_k)-\bar{\vS}_k}
    =O_p(r_Dn_k^{-1/2}+x_N/n_k).
\]
For \(\norm{\vt}\le1\) and \(\norm{\mE}_{\rm op}\le1/2\),
\[
\begin{aligned}
&\Phi_{k,n}(\vt,\mE)
 =\Phi_{k,n}(\bm0,\mE)-\widehat{\mG}_{k,n}(\mE)\vt
    +\mR_{k,n}(\vt,\mE),\cr
&\norm{\mR_{k,n}(\vt,\mE)}
\le C\norm{\vt}^2 n_k^{-1}\sum_i\norm{\vY_{ki}}^{-2}
   +C n_k^{-1}\sum_i\ind\{\norm{\vY_{ki}}\le C\norm{\vt}\}.
\end{aligned}
\]
Furthermore,
\[
    n_k^{-1}\sum_i\norm{\vY_{ki}}^{-2}=O_p(M_{Y,n}p^{-1}),
\]
\[
    n_k^{-1}\sum_i\ind\{\norm{\vY_{ki}}\le C\norm{\vt}\}
    =O_p(M_{Y,n}\norm{\vt}^2p^{-1}+x_N/n_k),
\]
and hence
\[
    \norm{\mR_{k,n}(\vt,\mE)}
    =O_p(M_{Y,n}\norm{\vt}^2p^{-1}+x_N/n_k).
\]
Lemma \ref{lem:matrix-bern-L} and the preceding scale expansion imply
\[
\begin{aligned}
\norm{\widehat{\mG}_{k,n}(\widehat{\mE}_k)-\mG_k}_{\rm op}
&\le
\left\|n_k^{-1}\sum_i\{\mL(\vY_{ki})-\mG_k\}\right\|_{\rm op}
 +\left\|n_k^{-1}\sum_i
       [\mL\{(\mI_p+\widehat{\mE}_k)\vY_{ki}\}(\mI_p+\widehat{\mE}_k)
        -\mL(\vY_{ki})]\right\|_{\rm op}\cr
&=O_p\left\{\mathfrak C_n^{1/2}p^{-1/2}(x_N/n_k)^{1/2}
      +\mathfrak C_np^{-1/2}r_D\right\}.
\end{aligned}
\]
Also
\[
    \norm{\bar{\vS}_k}=O_p(n_k^{-1/2}),
    \qquad
    \norm{\mG_k^{-1}}_{\rm op}\le p^{1/2}\mathfrak C_n^{1/4}.
\]
The equation \(\Phi_{k,n}(\vdelta_k,\widehat{\mE}_k)=\bm0\) gives
\[
\begin{aligned}
\vdelta_k-\mG_k^{-1}\bar{\vS}_k
={}&\mG_k^{-1}\{\Phi_{k,n}(\bm0,\widehat{\mE}_k)-\bar{\vS}_k\}
   -\mG_k^{-1}\{\widehat{\mG}_{k,n}(\widehat{\mE}_k)-\mG_k\}\vdelta_k\cr
&\quad+\mG_k^{-1}\mR_{k,n}(\vdelta_k,\widehat{\mE}_k).
\end{aligned}
\]
The same display and the lower bound
\(\lambda_{\min}(\mG_k)\ge p^{-1/2}\mathfrak C_n^{-1/4}\) yield
\[
    \norm{\vdelta_k}
    =O_p\{\mathfrak C_n^{1/4}(p/n_k)^{1/2}+q_{\theta,n}\}.
\]
Therefore
\[
\begin{aligned}
\norm{\vdelta_k-\mG_k^{-1}\bar{\vS}_k}
&\le
  O_p\{p^{1/2}\mathfrak C_n^{1/4}(r_Dn_k^{-1/2}+x_N/n_k)\}\cr
&\quad+O_p\left[p^{1/2}\mathfrak C_n^{1/4}
      \{\mathfrak C_n^{1/2}p^{-1/2}(x_N/n_k)^{1/2}
      +\mathfrak C_np^{-1/2}r_D\}
      \{\mathfrak C_n^{1/4}(p/n_k)^{1/2}+q_{\theta,n}\}\right]\cr
&\quad+O_p\left[p^{1/2}\mathfrak C_n^{1/4}
     \{M_{Y,n}\norm{\vdelta_k}^2p^{-1}+x_N/n_k\}\right]\cr
&=O_p(q_{\theta,n}).
\end{aligned}
\]
\end{proof}

\begin{lemma}\label{lem:sign-pert}
Under Assumptions \ref{ass:np}--\ref{ass:radius},
\[
    \max_k n_k^{-1}\sum_{i=1}^{n_k}\norm{\widehat{\vS}_{ki}-\vS_{ki}}
    =O_p(a_{S,n}).
\]
\end{lemma}

\begin{proof}
Let
\[
    \widehat{\mE}_k=\widehat{\mD}_k^{-1/2}\mD_k^{1/2}-\mI_p,
    \qquad
    \vdelta_k=\mD_k^{-1/2}(\widehat{\vtheta}_k-\vtheta_k).
\]
Then
\[
    \widehat{\vY}_{ki}
    =(\mI_p+\widehat{\mE}_k)(\vY_{ki}-\vdelta_k),
    \qquad
    \widehat{\vS}_{ki}=U\{(\mI_p+\widehat{\mE}_k)(\vY_{ki}-\vdelta_k)\}.
\]
On \(\mathcal E_{D,n}(M)\),
\[
    \norm{\widehat{\mE}_k}_{\rm op}=O(r_D),
    \qquad
    \norm{(\mI_p+\widehat{\mE}_k)\vY_{ki}}
    \ge (1-Cr_D)\norm{\vY_{ki}}.
\]
Lemma \ref{lem:sign-calculus} gives
\[
\begin{aligned}
\norm{\widehat{\vS}_{ki}-\vS_{ki}}
&\le
\norm{U\{(\mI_p+\widehat{\mE}_k)(\vY_{ki}-\vdelta_k)\}
       -U\{(\mI_p+\widehat{\mE}_k)\vY_{ki}\}}\cr
&\quad+\norm{U\{(\mI_p+\widehat{\mE}_k)\vY_{ki}\}-U(\vY_{ki})}\cr
&\le C\min\{1,\norm{\vdelta_k}/\norm{\vY_{ki}}\}+C\norm{\widehat{\mE}_k}_{\rm op}.
\end{aligned}
\]
Moreover,
\[
\begin{aligned}
\E\left(n_k^{-1}\sum_i\norm{\vY_{ki}}^{-1}\right)&\le M_{Y,n}p^{-1/2},\cr
\Var\left(n_k^{-1}\sum_i\norm{\vY_{ki}}^{-1}\right)&\le n_k^{-1}\E\norm{\vY_{k1}}^{-2}
      \le M_{Y,n}n_k^{-1}p^{-1},
\end{aligned}
\]
and hence
\[
    n_k^{-1}\sum_i\norm{\vY_{ki}}^{-1}=O_p(M_{Y,n}p^{-1/2}).
\]
By Lemma \ref{lem:median-rate},
\[
    \norm{\vdelta_k}
    \le \norm{\mG_k^{-1}\bar{\vS}_k}+O_p(q_{\theta,n})
    =O_p\{\mathfrak C_n^{1/4}(p/n_k)^{1/2}+q_{\theta,n}\}.
\]
Thus
\[
\begin{aligned}
 n_k^{-1}\sum_i\min\{1,\norm{\vdelta_k}/\norm{\vY_{ki}}\}
&\le \norm{\vdelta_k}\, n_k^{-1}\sum_i\norm{\vY_{ki}}^{-1}\cr
&=O_p\{\mathfrak C_n n_k^{-1/2}+\mathfrak C_n x_N/n_0\}.
\end{aligned}
\]
Consequently,
\[
\begin{aligned}
 n_k^{-1}\sum_i\norm{\widehat{\vS}_{ki}-\vS_{ki}}
&\le C\norm{\widehat{\mE}_k}_{\rm op}
   +C n_k^{-1}\sum_i\min\{1,\norm{\vdelta_k}/\norm{\vY_{ki}}\}\cr
&=O_p\{r_D+\mathfrak C_n n_k^{-1/2}+\mathfrak C_n x_N/n_0\}
 =O_p(a_{S,n}).
\end{aligned}
\]
\end{proof}

\begin{proof}[Proof of Proposition \ref{prop:derived}]
The expansion \eqref{eq:Bahadur-main} is Lemma \ref{lem:median-rate}.  The perturbation bound \eqref{eq:sign-pert} is Lemma \ref{lem:sign-pert}.

For \(\widehat{\mG}_k\),
\[
    \widehat{\mG}_k-\mG_k
    =n_k^{-1}\sum_i\{L(\vY_{ki})-\mG_k\}
     +n_k^{-1}\sum_i\{\mL(\widehat{\vY}_{ki})-L(\vY_{ki})\}.
\]
Lemma \ref{lem:matrix-bern-L} gives
\[
    \left\|n_k^{-1}\sum_i\{L(\vY_{ki})-\mG_k\}\right\|_{\rm op}
    =O_p\{\mathfrak C_n^{1/2}p^{-1/2}(x_N/n_k)^{1/2}\}.
\]
Lemma \ref{lem:sign-calculus}, \eqref{eq:inverse-radius}, \eqref{eq:smallball}, and Lemma \ref{lem:sign-pert} give
\[
    \left\|n_k^{-1}\sum_i\{\mL(\widehat{\vY}_{ki})-L(\vY_{ki})\}\right\|_{\rm op}
    =O_p(\mathfrak C_np^{-1/2}a_{S,n}).
\]
Hence
\[
    \norm{\widehat{\mG}_k-\mG_k}_{\rm op}
    =O_p\{\mathfrak C_n^{1/2}p^{-1/2}(x_N/n_k)^{1/2}
          +\mathfrak C_np^{-1/2}a_{S,n}\}.
\]
Hence
\[
    \norm{\widehat{\mG}_k^{-1}-\mG_k^{-1}}_{\rm op}
    \le
    \norm{\mG_k^{-1}}_{\rm op}^2\norm{\widehat{\mG}_k-\mG_k}_{\rm op}
    \{1+O_p(a_{K,n})\}.
\]
Together with \eqref{eq:D-event},
\[
    \max_{r=1,2,3}\norm{\widehat{\mK}_r-\mK_r}_{\rm op}=O_p(a_{K,n}).
\]
This proves \eqref{eq:Khat-rate}.  The plug-in covariance matrices \(\widehat{\mOmega}_k\) are still used in the algorithm, but the theory below calibrates the diagonal-deleted quadratic form through \(\widehat{\mH}_{n,0}\).  Thus no Frobenius consistency of the full empirical covariance matrix is required in regimes where \(p\) is larger than \(n_0\).
\end{proof}

\subsection{Expansion of the full-sample statistic}
\label{app:expansion}

\begin{proof}[Proof of Theorem \ref{thm:feasible}]
Under \(H_0\), \(\vtheta_1=\vtheta_2=\vtheta\).  Let
\[
    \vdelta_k=\mD_k^{-1/2}(\widehat{\vtheta}_k-\vtheta)
    =\mG_k^{-1}\bar{\vS}_k+\vDelta_k,
    \qquad
    \max_k\norm{\vDelta_k}=O_p(q_{\theta,n}).
\]
The first-order sign expansion used for \(\hat R_n^{\rm PDQ}\) gives, uniformly over the two samples,
\[
\begin{aligned}
U\{\widehat{\mD}_1^{-1/2}(\vX_{1i}-\widehat{\vtheta}_2)\}
&=\vS_{1i}-\mL(\vY_{1i})\mA_{12}\vdelta_2+\vr_{1i},\\
U\{\widehat{\mD}_2^{-1/2}(\vX_{2j}-\widehat{\vtheta}_1)\}
&=\vS_{2j}-\mL(\vY_{2j})\mA_{21}\vdelta_1+\vr_{2j},
\end{aligned}
\]
where
\[
    n_k^{-1}\sum_i\norm{\vr_{ki}}
    =O_p\{r_Dn_0^{-1/2}+x_Nn_0^{-3/2}\}.
\]
Deleting the empirical diagonal terms before bounding remainders gives
\[
\begin{aligned}
T_n^{\rm PDQ}
={}&\frac1{n_1^2}\sum_{i\ne \ell}\vS_{1i}\T\mK_1\vS_{1\ell}
      +\frac1{n_2^2}\sum_{j\ne \ell}\vS_{2j}\T\mK_2\vS_{2\ell}
      -\frac1{n_1n_2}\sum_{i,j}\vS_{1i}\T\mK_3\vS_{2j}\\
&\quad+O_p\{(r_D+a_{K,n}+L_{K,n}a_{S,n})\tau_n\}+O_p(x_N/n_0^2).
\end{aligned}
\]
The leading term is \(U_n\).  Therefore
\[
    T_n^{\rm PDQ}-U_n=O_p\{(r_D+a_{K,n}+L_{K,n}a_{S,n})\tau_n\}+O_p(x_N/n_0^2),
\]
which is \eqref{eq:main-expansion}.  The distributional conclusion \eqref{eq:full-law} follows from Theorem \ref{thm:oracle} and Assumption \ref{ass:spectral}.
\end{proof}

\subsection{Oracle quadratic-form law}
\label{app:oracle}

\begin{proof}[Proof of Theorem \ref{thm:oracle}]
The statistic \(U_n\) in \eqref{eq:Un} is a centered canonical quadratic statistic.  Its variance is
\[
\begin{aligned}
    \Var(U_n)
    ={}&\frac{2(n_1-1)}{n_1^3}
    \tr\{(\mOmega_1^{1/2}\mK_1\mOmega_1^{1/2})^2\}\\
    &+\frac{2(n_2-1)}{n_2^3}
    \tr\{(\mOmega_2^{1/2}\mK_2\mOmega_2^{1/2})^2\}
    +\frac1{n_1n_2}\tr(\mOmega_1\mK_3\mOmega_2\mK_3\T).
\end{aligned}
\]
Since \(n_k^{-1}(n_k-1)=1+O(n_0^{-1})\),
\[
    \Var(U_n)/\tau_n^2=1+O(n_0^{-1}).
\]
For the zero-diagonal kernel matrix of \(U_n\), the row sums satisfy
\[
\begin{aligned}
\max_a\sum_b\mathcal A_{n,0,ab}^2
\le C\max\Bigg\{&
\frac{\norm{\mK_1\mOmega_1\mK_1}_{\rm op}}{n_1^3},
\frac{\norm{\mK_2\mOmega_2\mK_2}_{\rm op}}{n_2^3},\\
&\frac{\norm{\mK_3\mOmega_2\mK_3\T}_{\rm op}}{n_1^2n_2},
\frac{\norm{\mK_3\T\mOmega_1\mK_3}_{\rm op}}{n_1n_2^2}
\Bigg\},
\end{aligned}
\]
and hence
\[
    \frac{\max_a\sum_b\mathcal A_{n,0,ab}^2}{\Var(U_n)}
    \le C\Delta_{{\rm row},n}\{1+O(n_0^{-1})\}.
\]
Let \(V_{1i}=\vS_{1i}\), \(V_{2j}=\vS_{2j}\), with block kernels induced by \(U_n\).  Lemma \ref{lem:dejong-tool} gives
\[
    h_n\le C\Delta_{{\rm row},n},
    \qquad
    d_n\le C\mathfrak C_n\Delta_{{\rm row},n}=C\delta_{{\rm dJ},n}.
\]
Therefore
\[
    \sup_t\abs{\Prb(U_n/\tau_n\le t)-\Prb(U_{n,G}/\tau_n\le t)}
    \le \omega_{\rm dJ}\{C\delta_{{\rm dJ},n}+C\Delta_{{\rm row},n}\}+O(n_0^{-1}).
\]
If
\[
    \vZ=(\vZ_1\T,\vZ_2\T)\T\sim N\{0,\diag(\mOmega_1,\mOmega_2)\},
\]
then the associated full Gaussian quadratic form is
\[
    \vZ\T\mC_n\vZ-\tr(\mB_n),
    \qquad
    \mC_n=
    \begin{pmatrix}
        n_1^{-1}\mK_1&-(2\sqrt{n_1n_2})^{-1}\mK_3\\
        -(2\sqrt{n_1n_2})^{-1}\mK_3\T&n_2^{-1}\mK_2
    \end{pmatrix}.
\]
The difference between \(U_{n,G}\) and this full centered Gaussian quadratic form is the Gaussian diagonal correction.  Its variance is bounded by
\[
    C\tau_n^2 n_0^{-1},
\]
using the same fourth-projection bound in \(\mathfrak C_n\).  Hence
\[
    \frac{U_{n,G}-\{\vZ\T\mC_n\vZ-\tr(\mB_n)\}}{\tau_n}=O_p(n_0^{-1/2}).
\]
Diagonalization of \(\mB_n\) gives
\[
    \vZ\T \mC_n\vZ-\tr(\mB_n)
    \stackrel{d}{=}
    \sum_{r=1}^{2p}\lambda_{nr}(\chi_{1r}^2-1)=\Gamma_n.
\]
Thus \eqref{eq:oracle-law} holds.  The \(\ell_2\) convergence in \eqref{eq:l2} yields \eqref{eq:limit-mixture}.
\end{proof}

\begin{proof}[Proof of Corollary \ref{cor:normal}]
The no-dominant-eigenvalue condition gives
\[
    \Delta_{{\rm BE},n}
    =\frac{\max_r\abs{\lambda_{nr}}}{(\sum_r\lambda_{nr}^2)^{1/2}}
    \to0.
\]
Lemma \ref{lem:be-chisq}, applied with \(a_{nr}=\lambda_{nr}\), gives
\[
    \sup_t\abs{\Prb(\Gamma_n/\tau_n\le t)-\Phi(t)}
    \le C\Delta_{{\rm BE},n}.
\]
Theorem \ref{thm:feasible} transfers the normal limit to \((\hat R_n^{\rm PDQ}-b_n)/\tau_n\).
\end{proof}

\subsection{Diagonal-deletion correction}
\label{app:bhat}

The identity \eqref{eq:bhat} gives
\[
\hat b_n=\sum_{k=1}^2n_k^{-2}\sum_{i=1}^{n_k}
       \widehat{\vS}_{ki}\T\widehat{\mK}_k\widehat{\vS}_{ki}.
\]
Thus subtracting \(\hat b_n\) from the full-sample quadratic expansion removes exactly the fitted within-sample diagonal terms:
\[
\begin{aligned}
T_n^{\rm PDQ}
={}&\frac1{n_1^2}\sum_{i\ne \ell}\widehat{\vS}_{1i}\T\widehat{\mK}_1\widehat{\vS}_{1\ell}
  +\frac1{n_2^2}\sum_{j\ne \ell}\widehat{\vS}_{2j}\T\widehat{\mK}_2\widehat{\vS}_{2\ell}\\
&\quad-\frac1{n_1n_2}\sum_{i=1}^{n_1}\sum_{j=1}^{n_2}
\widehat{\vS}_{1i}\T\widehat{\mK}_3\widehat{\vS}_{2j}
+O_p\{(r_D+a_{K,n}+L_{K,n}a_{S,n})\tau_n\}+O_p(x_N/n_0^2),
\end{aligned}
\]
where the remainder is the same as in the proof of Theorem \ref{thm:feasible}.  This diagonal-deletion is the reason no condition of the form \(\hat b_n-b_n=o_p(\tau_n)\) is needed.

\subsection{Wild bootstrap}
\label{app:boot}

\begin{proof}[Proof of Proposition \ref{prop:boot-rates}]
For a row corresponding to the first sample,
\[
\begin{aligned}
\sum_{b\ne 1i}\hat a_{1i,b}^2
&\le \frac1{n_1^4}\sum_{\ell\ne i}
  (\widehat{\vS}_{1i}\T\widehat{\mK}_1\widehat{\vS}_{1\ell})^2
 +\frac1{4n_1^2n_2^2}\sum_{j=1}^{n_2}
  (\widehat{\vS}_{1i}\T\widehat{\mK}_3\widehat{\vS}_{2j})^2  \\
&\le C\left\{
\frac{\norm{\widehat{\mK}_1\widehat{\mOmega}_1\widehat{\mK}_1}_{\rm op}}{n_1^3}
+\frac{\norm{\widehat{\mK}_3\widehat{\mOmega}_2\widehat{\mK}_3\T}_{\rm op}}{n_1^2n_2}
\right\}.
\end{aligned}
\]
The second-sample rows are bounded analogously.  By \eqref{eq:sign-pert} and \eqref{eq:Khat-rate},
\[
\begin{aligned}
\max\{&
\norm{\widehat{\mK}_1\widehat{\mOmega}_1\widehat{\mK}_1}_{\rm op},
\norm{\widehat{\mK}_2\widehat{\mOmega}_2\widehat{\mK}_2}_{\rm op},\\
&\norm{\widehat{\mK}_3\widehat{\mOmega}_2\widehat{\mK}_3\T}_{\rm op},
\norm{\widehat{\mK}_3\T\widehat{\mOmega}_1\widehat{\mK}_3}_{\rm op}
\}
\end{aligned}
\]
\[
\le
\max\{
\norm{\mK_1\mOmega_1\mK_1}_{\rm op},
\norm{\mK_2\mOmega_2\mK_2}_{\rm op},
\norm{\mK_3\mOmega_2\mK_3\T}_{\rm op},
\norm{\mK_3\T\mOmega_1\mK_3}_{\rm op}
\}
+O_p\{\tau_n^2(a_{K,n}^2+L_{K,n}^2a_{S,n}^2)\}.
\]
Therefore \eqref{eq:boot-row} follows from \eqref{eq:Delta-row}.  For the conditional variance,
\[
\begin{aligned}
2\tr(\widehat{\mH}_{n,0}^2)
={}&\frac{2}{n_1^4}\sum_{i\ne \ell}
(\widehat{\vS}_{1i}\T\widehat{\mK}_1\widehat{\vS}_{1\ell})^2
+\frac{2}{n_2^4}\sum_{j\ne \ell}
(\widehat{\vS}_{2j}\T\widehat{\mK}_2\widehat{\vS}_{2\ell})^2\\
&\quad+\frac1{n_1^2n_2^2}\sum_{i=1}^{n_1}\sum_{j=1}^{n_2}
(\widehat{\vS}_{1i}\T\widehat{\mK}_3\widehat{\vS}_{2j})^2.
\end{aligned}
\]
The order-two U-statistic concentration applied to the three displayed sums gives
\[
    \frac{2\tr(\widehat{\mH}_{n,0}^2)}{\tau_n^2}
    =1+O_p\{a_{K,n}^2+L_{K,n}^2a_{S,n}^2+\Delta_{{\rm row},n}+n_0^{-1}\},
\]
which is \eqref{eq:boot-var}.  Equation \eqref{eq:hattau-star-rate} follows from
\[
    \abs{\sqrt{1+x}-1}\le C\abs{x},\qquad \abs{x}\le1/2,
\]
and the preceding display.
\end{proof}

\begin{proof}[Proof of Theorem \ref{thm:boot}]
Since \(e_a^2\equiv1\),
\[
    T_n^*=Q_n^*-\hat b_n
    =\sum_{a\ne b}\hat a_{ab}e_ae_b
    =\ve\T\widehat{\mH}_{n,0}\ve.
\]
Thus
\[
    \Estar T_n^*=0,
    \qquad
    \Varstar(T_n^*)=2\tr(\widehat{\mH}_{n,0}^2).
\]
Let \(g_1,\ldots,g_N\) be conditionally independent \(N(0,1)\) variables and write \(\vg=(g_1,\ldots,g_N)\T\).  Set
\[
    T_{n,G}^*=\vg\T\widehat{\mH}_{n,0}\vg.
\]
The conditional row influence satisfies, by Proposition \ref{prop:boot-rates},
\[
    h_n^*=\frac{\max_a\sum_{b\ne a}\hat a_{ab}^2}{\tr(\widehat{\mH}_{n,0}^2)}
    =O_p(\Delta_{{\rm row},n}+a_{K,n}^2+L_{K,n}^2a_{S,n}^2).
\]
Lemma \ref{lem:dejong-tool}, conditionally on \(\calX\), gives
\[
    \sup_t\abs{\Prstar(T_n^*/\tau_n\le t)-\Prstar(T_{n,G}^*/\tau_n\le t)}
    =O_p\{\omega_{\rm dJ}(C\delta_{H,n})\}.
\]
Let \(\hat\nu_{nr}\) be the eigenvalues of \(\widehat{\mH}_{n,0}\).  Conditional diagonalization gives
\[
    T_{n,G}^*\stackrel{d^*}{=}\sum_r\hat\nu_{nr}(\chi_{1r}^2-1).
\]
The variance and kernel perturbation bounds in Proposition \ref{prop:boot-rates}, together with Lemma \ref{lem:HW}, imply
\[
    \sum_r(\hat\nu_{nr}-\lambda_{nr})^2=O_p(\delta_{H,n}\tau_n^2).
\]
Since \(\E(\chi_1^2-1)^2=2\),
\[
    \frac{\sum_r(\hat\nu_{nr}-\lambda_{nr})(\chi_{1r}^2-1)}{\tau_n}
    =O_p(\delta_{H,n}^{1/2})
    \quad\hbox{in }L^2.
\]
Therefore
\[
    \sup_t\abs{\Prstar(T_n^*/\tau_n\le t)-\Prb(\Gamma_n/\tau_n\le t)}\to0
\]
in probability.  Theorem \ref{thm:feasible} gives
\[
    \frac{T_n^{\rm PDQ}}{\tau_n}=\frac{U_n}{\tau_n}
     +O_p\left(r_D+a_{K,n}+L_{K,n}a_{S,n}+\frac{x_N}{n_0^2\tau_n}\right),
\]
and Theorem \ref{thm:oracle} gives
\[
    \sup_t\abs{\Prb(T_n^{\rm PDQ}/\tau_n\le t)-\Prb(\Gamma_n/\tau_n\le t)}\to0.
\]
Combining the last two displays proves \eqref{eq:boot-main}.  Proposition \ref{prop:boot-rates} and Slutsky's theorem give \eqref{eq:boot-main-hattau}.  The size statement follows from bootstrap quantile consistency at continuity points of the limiting null law.
\end{proof}

\section{Comments on the radius and spectral assumptions}
\label{app:assumption-comments}

This appendix records several elementary implications of the elliptical model that are used to interpret Assumptions \ref{ass:radius} and \ref{ass:spectral}.

\subsection{Pairwise-quantile scaling and the working diagonal}

Let \(\ve_j\) denote the \(j\)th coordinate vector.  Under \eqref{eq:model} {\color{black} and $H_0$},
\[
        X_{ki,j}-X_{k\ell,j}
        =(d_{kj}^{\circ})^{1/2}
        \ve_j\T\mR_k^{1/2}
        (\xi_{ki}\vu_{ki}-\xi_{k\ell}\vu_{k\ell}).
\]
Since \(\mR_k\) is a correlation matrix,
\[
        \norm{\mR_k^{1/2}\ve_j}^2
        =\ve_j\T\mR_k\ve_j=1.
\]
Rotational invariance of \(\vu_{ki}\sim {\rm Unif}(\mathbb S^{p-1})\) yields
\[
        \ve_j\T\mR_k^{1/2}\xi_{ki}\vu_{ki}
        \stackrel{d}{=}\xi_{ki}U_{ki,1},
\]
where \(U_{ki,1}\) is the first coordinate of a uniform random vector on \(\mathbb S^{p-1}\).  Thus, with
\[
        q_{k,\alpha}
        =
        Q_{\alpha}\left(
        \abs{\xi_{k1}U_{k1,1}-\xi_{k2}U_{k2,1}}
        \right),
\]
we have, for every coordinate \(j\),
\[
        Q_{\alpha}\{\abs{X_{k1,j}-X_{k2,j}}\}
        =
        q_{k,\alpha}(d_{kj}^{\circ})^{1/2}.
\]
Consequently the population target of the pairwise-quantile diagonal estimator is
\[
        d_{kj}=q_{k,\alpha}^2d_{kj}^{\circ},
        \qquad
        \mD_k=q_{k,\alpha}^2\mD_k^{\circ}.
\]
The unknown scalar \(q_{k,\alpha}\) is immaterial for spatial signs because, for any \(c>0\),
\[
        U\{(c\mD_k)^{-1/2}\vx\}=U(\mD_k^{-1/2}\vx).
\]
In particular,
\[
        \vY_{ki}
        =
        \mD_k^{-1/2}(\vX_{ki}-\vtheta_k)
        =
        q_{k,\alpha}^{-1}\xi_{ki}\mR_k^{1/2}\vu_{ki},
\]
and
\[
        \norm{\vY_{ki}}^2
        =
        q_{k,\alpha}^{-2}\xi_{ki}^2
        B_{ki},
        \qquad
        B_{ki}=\vu_{ki}\T\mR_k\vu_{ki}.
\]

\subsection{Interpretation of Assumption \ref{ass:radius}}

From the preceding display, for \(r=1,2,4\),
\[
        p^{r/2}\E\norm{\vY_{k1}}^{-r}
        =
        p^{r/2}q_{k,\alpha}^r
        \E(\xi_{k1}^{-r})\E(B_{k1}^{-r/2}).
\]
Thus Assumption \ref{ass:radius} separates a one-dimensional radial lower-tail requirement from the shape contribution \(\E(B_{k1}^{-r/2})\).

For the AR(1) shape \((\mR_k)_{ab}=\rho_k^{\abs{a-b}}\), \(\abs{\rho_k}\le\bar\rho<1\),
\[
        \frac{1-\bar\rho}{1+\bar\rho}
        \le B_{k1}\le
        \frac{1+\bar\rho}{1-\bar\rho},
        \qquad
        \E(B_{k1}^{-r/2})=O(1).
\]
Moreover,
\[
        \E B_{k1}=1,
        \qquad
        \Var(B_{k1})
        =
        \frac{2}{p(p+2)}
        \{\tr(\mR_k^2)-p\}
        =O(p^{-1}),
\]
and hence \(B_{k1}=1+O_p(p^{-1/2})\).

For the compound-symmetry shape
\[
        \mR_k=(1-\rho_k)\mI_p+\rho_k\bm{1}_p\bm{1}_p\T,
        \qquad
        0\le\rho_k\le\bar\rho<1,
\]
write \(\bm v_p=p^{-1/2}\bm{1}_p\).  Then
\[
        B_{k1}
        =
        1-\rho_k+\rho_k p(\bm v_p\T\vu_{k1})^2,
        \qquad
        (\bm v_p\T\vu_{k1})^2
        \sim {\rm Beta}\{1/2,(p-1)/2\}.
\]
Therefore
\[
        B_{k1}\ge1-\bar\rho,
        \qquad
        \E(B_{k1}^{-r/2})=O(1),
        \qquad
        B_{k1}=O_p(1).
\]
Also, for both AR(1) and fixed-\(\rho\) compound symmetry,
\[
        \Prb\{\norm{\vY_{k1}}\le t p^{1/2}\}
        \le
        \Prb\{\xi_{k1}\le Cq_{k,\alpha}t p^{1/2}\},
        \qquad
        0<t\le1.
\]
Thus Assumption \ref{ass:radius} reduces in these benchmark cases to a one-dimensional lower-tail condition on the radial variable.  If \(\rho_k\to1\) in the compound-symmetry model, the lower bound on \(B_{k1}\) degenerates and \(M_{Y,n}\) may have to diverge.

\subsection{Relation with the inverse-radius condition of \citet{ZouPengFengWang2014}}

\citet{ZouPengFengWang2014} imposed, in their notation with \(R_i=\norm{\vX_i-\vtheta}\),
\[
        \frac{\E(R_i^{-r})}{\{\E(R_i^{-1})\}^{r}}
        \to d_r\in[1,\infty),
        \qquad
        r=2,3,4.
\]
Assumption \ref{ass:radius} is a rate-explicit analogue for the standardized radius \(R_{k1}^{Y}=\norm{\vY_{k1}}\).  Indeed, \eqref{eq:inverse-radius} gives
\[
        p^{r/2}\E\{(R_{k1}^{Y})^{-r}\}
        \le C M_{Y,n},
        \qquad r=1,2,4,
\]
and the case \(r=3\) follows from Cauchy's inequality:
\[
        p^{3/2}\E\{(R_{k1}^{Y})^{-3}\}
        \le
        \{p\E(R_{k1}^{Y})^{-2}\}^{1/2}
        \{p^2\E(R_{k1}^{Y})^{-4}\}^{1/2}
        \le M_{Y,n}.
\]
If \(p^{1/2}\E(R_{k1}^{Y})^{-1}\) is bounded away from zero, then
\[
        \frac{\E\{(R_{k1}^{Y})^{-r}\}}
        {\{\E(R_{k1}^{Y})^{-1}\}^{r}}
        =O(M_{Y,n}),
        \qquad r=2,3,4.
\]
Conversely, the condition of \citet{ZouPengFengWang2014}, together with
\[
        \E\{(R_{k1}^{Y})^{-1}\}\asymp p^{-1/2},
\]
implies the inverse-moment part of \eqref{eq:inverse-radius} with bounded \(M_{Y,n}\).  The additional small-ball condition \eqref{eq:smallball} is used to control the empirical spatial-median Jacobian uniformly; it is a lower-tail condition and does not require positive moments of the radial variable.

\subsection{Common elliptical models satisfying Assumption \ref{ass:radius}}

Let \(\vZ\sim N_p(\mathbf 0,\mI_p)\) and \(R_p=\norm{\vZ}\).  For fixed \(r<p\),
\[
        \E R_p^{-r}
        =
        2^{-r/2}
        \frac{\Gamma\{(p-r)/2\}}{\Gamma(p/2)}
        \asymp p^{-r/2}.
\]
Also, for every fixed \(c>0\),
\[
        \sup_{0<t\le1}t^{-2}
        \Prb(R_p\le ctp^{1/2})
        \le C_c.
\]
Consequently the multivariate normal model, for which \(\xi_{ki}=R_p\), satisfies \eqref{eq:inverse-radius} and \eqref{eq:smallball} with \(M_{Y,n}=O(1)\) whenever \(q_{k,\alpha}\) is bounded away from zero and infinity.

For the multivariate \(t_{\nu}\) model with fixed \(\nu>2\),
\[
        \xi_{ki}=\frac{R_p}{(W/\nu)^{1/2}},
        \qquad
        W\sim\chi_{\nu}^2,
        \qquad
        W\perp R_p.
\]
Hence, for fixed \(r\le4\),
\[
        \E\xi_{ki}^{-r}
        =
        \E R_p^{-r}\,
        \E(W/\nu)^{r/2}
        \asymp p^{-r/2}.
\]
Choose \(\beta\in(2/\nu,1)\).  For \(p\) sufficiently large,
\[
\begin{aligned}
        \Prb(\xi_{ki}\le ctp^{1/2})
        &\le
        \Prb\{R_p\le ct^{1-\beta}p^{1/2}\}
        +\Prb\{W/\nu>t^{-2\beta}\}  \\
        &\le C t^2,
        \qquad 0<t\le1.
\end{aligned}
\]
Thus the fixed-degree \(t_{\nu}\) model, including \(t_3\), satisfies Assumption \ref{ass:radius}.

For a finite-scale mixture normal model,
\[
        \xi_{ki}=S_{ki}^{1/2}R_p,
        \qquad
        0<s_-\le S_{ki}\le s_+<\infty,
\]
we have
\[
        \E\xi_{ki}^{-r}
        \le s_-^{-r/2}\E R_p^{-r}
        \asymp p^{-r/2},
\]
and
\[
        \Prb(\xi_{ki}\le ctp^{1/2})
        \le
        \Prb(R_p\le cs_-^{-1/2}tp^{1/2})
        \le C t^2.
\]
Therefore the finite-scale mixture normal model also satisfies Assumption \ref{ass:radius}.  These calculations justify the three simulation models used in Section \ref{sec:simulation}.

\subsection{Benchmark implications of Assumption \ref{ass:spectral}}

Assumption \ref{ass:spectral} is the only condition involving the shape structure.  It does not impose sparsity, bandedness, bounded maximum row sum, or weak correlation on \(\mR_k\).  Its row-leverage condition is expressed through the angular covariance matrices \(\mOmega_k\), rather than the crude bound \(L_{K,n}^2n_0^{-3}\tau_n^{-2}\).

Assume for simplicity that
\[
        n_1\asymp n_2\asymp n_0,
        \qquad
        \mathfrak C_n=O(1),
        \qquad
        L_{K,n}=O(1).
\]
Then
\[
        r_D\asymp a_{S,n}\asymp a_{K,n}
        \asymp \left(\frac{x_N}{n_0}\right)^{1/2},
        \qquad
        q_{\theta,n}\asymp \frac{p^{1/2}x_N}{n_0}.
\]

In the bounded-spectrum regime,
\[
        0<c\le \lambda_{\min}(\mR_k)
        \le \lambda_{\max}(\mR_k)\le C<\infty,
\]
one has
\[
        \mOmega_k\asymp p^{-1}\mI_p,
        \qquad
        \tau_n^2\asymp \frac{1}{n_0^2p},
        \qquad
        \tau_n\asymp \frac{1}{n_0p^{1/2}}.
\]
The angular row-leverage term satisfies
\[
        \Delta_{{\rm row},n}\asymp n_0^{-1},
        \qquad
        \frac{x_N}{n_0^2\tau_n}
        \asymp
        \frac{p^{1/2}x_N}{n_0},
        \qquad
        \delta_{H,n}\asymp \frac{x_N}{n_0}+\frac1{n_0}.
\]
Thus a convenient sufficient condition is
\[
        \frac{x_N}{n_0}\to0,
        \qquad
        \frac{p^{1/2}x_N}{n_0}\to0,
\]
or equivalently
\[
        p=o\left\{\frac{n_0^2}{\log^2(p+n_0)}\right\}.
\]
This includes the usual AR(1) correlation model with fixed \(\abs{\rho}<1\).

In the compound-symmetry regime,
\[
        \mR_k=(1-\rho_k)\mI_p+\rho_k\bm{1}_p\bm{1}_p\T,
        \qquad
        0<\underline\rho\le \rho_k\le \bar\rho<1,
\]
let
\[
        \bm v_p=p^{-1/2}\bm{1}_p,
        \qquad
        \mathbf P_{\parallel}=\bm v_p\bm v_p\T,
        \qquad
        \mathbf P_{\perp}=\mI_p-\mathbf P_{\parallel}.
\]
Then
\[
        \mR_k
        =
        \{1-\rho_k+\rho_kp\}\mathbf P_{\parallel}
        +(1-\rho_k)\mathbf P_{\perp},
\]
and, by symmetry,
\[
        \mOmega_k
        =
        \omega_{k1}\mathbf P_{\parallel}
        +\omega_{k2}\mathbf P_{\perp},
        \qquad
        \omega_{k1}\asymp1,
        \qquad
        \omega_{k2}\asymp p^{-1}.
\]
If the leading compound-symmetry direction is not annihilated by \(\mK_1,\mK_2,\mK_3\), then
\[
        \tau_n^2\asymp n_0^{-2},
        \qquad
        \tau_n\asymp n_0^{-1},
        \qquad
        \Delta_{{\rm row},n}\asymp n_0^{-1}.
\]
Consequently,
\[
        \frac{x_N}{n_0^2\tau_n}\asymp \frac{x_N}{n_0},
        \qquad
        \delta_{H,n}\asymp \frac{x_N}{n_0}+\frac1{n_0}.
\]
The binding plug-in requirement remains
\[
        q_{\theta,n}\to0,
        \qquad\text{that is,}\qquad
        \frac{p^{1/2}x_N}{n_0}\to0,
\]
which again gives
\[
        p=o\left\{\frac{n_0^2}{\log^2(p+n_0)}\right\}.
\]
Thus weakly dependent AR-type structures are not subject to the artificial \(p\log(p+n_0)=o(n_0)\) restriction produced by crude operator-norm row bounds.  Under compound symmetry the same dimensional range is allowed by the plug-in analysis, while the limiting law is typically non-Gaussian because the leading eigenvalue of \(\mB_n\) is non-negligible.

\bibliographystyle{apa}
\bibliography{ref}

\end{document}